\documentclass[runningheads]{llncs}
\usepackage[colorlinks,citecolor=green!60!black,linkcolor=red!60!black]{hyperref}
\usepackage{color}
\usepackage{multirow, array}
\usepackage{tabularx}
\usepackage{amsmath}
\usepackage{amssymb,amsfonts,textcomp}
\usepackage{latexsym}
\usepackage{graphicx} 
\usepackage{microtype}
\usepackage{float}
\usepackage{rotating}
\usepackage{url}
\usepackage{comment}
\usepackage{paralist}
\usepackage{booktabs}
\usepackage{bm}
\usepackage[]{units}
\usepackage[noend, ruled, noline]{algorithm2e}
\usepackage{tikz}
\usepackage[sort&compress,comma,numbers,sectionbib]{natbib}
\usepackage{mathtools}
\usepackage{cleveref}
\usepackage{rotating,tcolorbox,complexity}
\usepackage{thmtools}
\usepackage{thm-restate}

\usepackage{boxedminipage}
\usepackage{xspace}

\allowdisplaybreaks

\newcommand{\ie}{i.e.,\xspace}

\newcommand{\problemdef}[4]{
	\begin{center}
		\begin{minipage}{0.95\textwidth}
			\normalsize\textsc{#2} \smallskip \\
			\begin{tabularx}{\textwidth}{@{}l@{\hspace{3pt}}X}
				\normalsize\textbf{Input:} & \normalsize#3 \\
				\normalsize\textbf{#1:}    & \normalsize#4
			\end{tabularx}
		\end{minipage}
	\end{center}
}
\newcommand{\dprob}[4][Question]{\problemdef{#1}{#2}{#3}{#4}}
\newcommand{\optprob}[4][Task]{\problemdef{#1}{#2}{#3}{#4}}
\DeclareMathOperator{\vote}{vote}

\newenvironment{proofsketch}{%
  \proof}{\endproof}

\newcommand{\wone}{{\mathrm{W[1]}}}
\newcommand{\w}{{\mathrm{W}}}
\newcommand{\fpt}{{\mathrm{FPT}}}

\newcommand{\Mp}{M^\prime}
\newcommand{\hopen}{P_{\text{open}}}
\newcommand{\mquota}{m_{\text{quota}}}
\newcommand{\mopen}{m_{\text{open}}}

\newcommand{\mclosed}{m_{\text{closed}}}
\newcommand{\lmax}{l_{\text{max}}}
\newcommand{\umax}{u_{\text{max}}}
\newcommand{\degA}{\Delta_A}
\newcommand{\degH}{\Delta_P}
\newcounter{captionedequationset} 
\newdimen\captionlength
\newcommand{\captionedequationset}[1]{
   ~\refstepcounter{captionedequationset}%
    \setlength{\captionlength}{\widthof{#1}} 
    \addtolength{\captionlength}{\widthof{Equation set~\thecaptionedequationset: }}
    \ifthenelse{\lengthtest{\captionlength < \linewidth }} 
    {\begin{center}
            Preference structure~\thecaptionedequationset: #1
        \end{center}} 
    { \begin{flushleft} 
        Preference structure~\thecaptionedequationset: #1 
        \end{flushleft}}}
\newcommand{\popular}{\textsc{pop-ha$_L^U$\;}}
\newcommand{\popver}{\textsc{popv-ha$_L^U$\;}}
\newcommand{\maxpareto}{\textsc{perpo-ha$_L^U$\;}}
\newcommand{\parver}{\textsc{pov-ha$_L^U$\;}}
\newcommand{\wmatch}{\textsc{w-ha$_L^U$\;}}

\newcommand{\popularnsp}{\textsc{pop-ha$_L^U$}}
\newcommand{\popvernsp}{\textsc{popv-ha$_L^U$}}
\newcommand{\maxparetonsp}{\textsc{perpo-ha$_L^U$}}
\newcommand{\parvernsp}{\textsc{pov-ha$_L^U$}}
\newcommand{\wmatchnsp}{\textsc{w-ha$_L^U$}}
\usepackage[textsize=scriptsize,textwidth=3cm,shadow]{todonotes}

\title{Pareto optimal and popular house allocation\\ with lower and upper quotas} 

\titlerunning{Optimal house allocation with lower and upper quotas}

\begin{document}
\author{
\'Agnes Cseh\inst{1,2}%
\and
Tobias Friedrich \inst{1}%
\and
Jannik Peters\inst{3}%
}
\authorrunning{\'A. Cseh \and T. Friedrich \and J. Peters}
\institute{Hasso-Plattner-Institute, University of Potsdam, Germany \and
Institute of Economics, Centre for Economic and Regional Studies, Hungary\\
\email{agnes.cseh@hpi.de}, \email{tobias.friedrich@hpi.de}\\
\and TU Berlin, Germany,
\email{jannik.peters@tu-berlin.de}
}
\maketitle

\begin{abstract}
In the house allocation problem with lower and upper quotas, we are given a set of applicants and a set of projects. Each applicant has a strictly ordered preference list over the projects, while the projects are equipped with a lower and an upper quota. A feasible matching assigns the applicants to the projects in such a way that a project is either matched to no applicant or to a number of applicants between its lower and upper quota. 

In this model we study two classic optimality concepts: Pareto optimality and popularity. We show that finding a popular matching is hard even if the maximum lower quota is~$2$ and that finding a perfect Pareto optimal matching, verifying Pareto optimality, and verifying popularity are all \NP-complete even if the maximum lower quota is~$3$. We complement the last three negative results by showing that the problems become polynomial-time solvable when the maximum lower quota is~$2$, thereby answering two open questions of Cechl\'arov\'a and Fleiner~\cite{CF17}. Finally, we also study the parameterized complexity of all four mentioned problems.
\end{abstract}%

\section{Introduction}%

Many university courses involve team-based project work. In such courses, a set of projects is offered and each student submits a list of projects she %
finds acceptable. Ideally, the student also ranks these projects in order of her preference. Naturally, the number of students ideally assigned to a specific project strongly depends on the project itself. The lecturer responsible for the project thus might restrict the number of students to an interval. Projects that did not awake sufficient interest in the students are then dropped, while the other projects start with a number of assigned students that falls into the prescribed interval. Such quota constraints also arise in various other contexts involving the centralized formation of groups, including organizing team-based leisure activities, opening facilities to serve a community, and coordinating rides within car-sharing systems. In these and similar applications, the goal is to fulfill some optimality condition under the assumption that the number of participants for each open activity is within the prescribed limits of the activity.

\subsection{Problem formulation and solution concepts}

The mathematical formulation of this problem is known as the house allocation problem with lower and upper quotas. 
In a house allocation instance, we are given a two-sided market, where one side $A$ represents applicants, while the other side $P$ represents projects. 
Each applicant has a strictly ordered preference list of the projects she finds acceptable. Furthermore, each project $p \in P$ has a lower quota $\ell_p$ and an upper quota~$u_p$.

In a feasible matching, a project is either \emph{open} or \emph{closed}. The central feasibility requirement is that the number of applicants assigned to an open project must lie between its lower and upper quota, whilst a closed project has no assigned applicant. Each applicant is assigned to at most one project. We define the optimality of a matching with respect to the satisfaction level of the agents. In this paper, we study two well-known notions from the broad topic of matchings under preferences: Pareto optimality and popularity. 

A matching $M$ is \textit{Pareto optimal} if there is no matching $\Mp$, in which no applicant is matched to a project she considers worse, while at least one applicant is matched to a project she considers better than her project in~$M$.
A matching $M$ is \textit{popular} if there is no matching $\Mp$ that would win a head-to-head election against $M$, where each applicant casts a vote based on her preferences on her assigned project in $M$ and in~$M'$.

\subsection{Related work}

Arulselvan et al.~\cite{ACGMM18} derived several complexity results for the maximum weight many-to-one matching problem with project closures and lower and upper quotas. However, their model 
excludes agent preferences. In Table~\ref{ta:literature}
we display a structured overview of existing work in the field of matchings under preferences with lower and upper quotas.

\subsubsection{Stable matchings}
In the classic hospitals residents problem~\cite{GS62,GI89}, the underlying model is a bipartite many-to-one matching problem involving preferences on both sides, and the goal is to find a stable matching, which is a matching where no hospital-resident pair could improve their situation by being assigned to each other. This model has been combined with lower and upper quotas in several papers. %
Hamada et al.~\cite{HIM16} considered a version where hospitals cannot be closed and presented a polynomial-time algorithm to find a stable solution, while Mnich and Schlotter~\cite{MS20} studied the fixed parameter tractability of finding an approximately stable solution in no-instances identified by Hamada et al.~\cite{HIM16}. The model of Bir\'o et al.~\cite{BFIM10} permitted hospital closures, and was shown to lead to \NP-hardness. Very recently, the setting of Bir\'o et al.~\cite{BFIM10} was further investigated by Boehmer and Heeger~\cite{BH20}, who conducted a parameterized study in the original hospitals residents setting, and also in the house allocation setting, where hospitals only have a preference for filling their quota, but do not mind which applicant is assigned to them. They also answered an open question of Bir\'o et al.~\cite{BFIM10} by showing that a stable matching in the hospitals residents problem with lower quota at most 2 can be found in polynomial time. For a further overview on matchings with quotas and constraints we refer the reader to a recent survey by \citet{ABY21}. %
\begin{table}[tb]
	\begin{center}
		\begin{tabularx}{.85\textwidth}{lccc}
                \toprule			
			&  \textbf{Stability}  & \textbf{Pareto optimality} & \textbf{Popularity}  \\ 
        	\midrule
			1-sided, no project closures & open  & \cite{GHI+14} & open  \\ 
			1-sided, project closures &  \cite{BH20} & \cite{MT13,K13,CF17}, our paper & our paper\\ \hline
			2-sided, no project closures & \cite{HIM16, MS20} & \cite{SCAJ19}& \cite{ANNR18,NN18} \\ 
			2-sided, project closures &\cite{BFIM10,BH20} & open
			& open \\
			\bottomrule	
		\end{tabularx}
		\caption{Overview of the existing literature in the most related settings. The four models differ in how many of the two sides are equipped with preferences, and in the possibility of project closures. Note that stability is defined for one-sided preferences such that hospitals do not differentiate between applicants, but aim to fill their quota.}
		\label{ta:literature}
	\end{center}
\end{table}
\subsubsection{Pareto optimal matchings}
Pareto optimality is one of the most studied concepts in coalition formation and hedonic games~\cite{ABH13,BFO17,B20,EFF20}, and it has also been defined in the context of various matching markets~\cite{CEFM+14,CEFM+16,ACGS18,BG20}. 
As shown by Abraham et al.~\cite{ACMM04}, in the one-to-one house allocation model, a maximum size Pareto optimal matching can be found in polynomial time.
Pareto optimality of matchings with lower and upper quotas on projects was studied in four papers. Motivated by a school choice application with regional constraints, Goto et al.~\cite{GHI+14} analyzed the case of so-called hierarchical lower quotas that must be obeyed. Monte and Tumennasan~\cite{MT13} considered the case of  project closures with complete lists, while the model of Kamiyama~\cite{K13} allowed incomplete lists as well. In all three works, it was shown that
a Pareto optimal matching can always be found using a variant of the famous serial dictatorship algorithm. Cechl\'arov\'a and Fleiner~\cite{CF17} extended this algorithm to the case when an applicant can be assigned to more than one project. They also showed for the many-to-one case with lower and upper quotas that it is \NP-hard to compute a maximum size Pareto optimal matching if the maximum lower quota is at least $4$, furthermore that it is \NP-complete to verify if a matching is Pareto optimal if the maximum lower quota is at least $3$. This lead the authors to ask whether both of these problems stay intractable if no lower quota exceeds $2$. In this paper we show that both problems are indeed polynomial-time solvable if the maximum lower quota is~$2$, while finding a maximum size Pareto optimal matching is \NP-complete if the maximum lower quota is $3$. Regarding two-sided instances, Sanchez-Anguix et al.~\cite{SCAJ19} conducted experiments to derive an approximate Pareto optimal solution with workload balance as an additional requirement. 

\subsubsection{Popular matchings} Popularity as an optimality principle has been on the rise recently~\cite{Cse17,FKPZ19,GMSZ21} in the matchings under preferences literature. On instances with two-sided preferences, Brandl and Kavitha~\cite{BK19} and \citet{GNNR19} studied popularity for many-to-many and many-to-one matching problems with upper quotas only. %
For the model introduced in the latter paper, the complexity of %
deciding whether a popular matching exists is still open. %
Krishnapriya et al.~\cite{ANNR18} and Nasre and Nimbhorkar~\cite{NN18} investigated popular matchings in the hospital residents problem with lower and upper quotas, but without the option to close hospitals. 
They proved that whenever a feasible matching exists, a popular matching has to exist as well. 
In the house allocation setting, the problem of computing a popular matching is tractable, even if both upper quotas and applicant weights are present, as shown by Sng and Manlove~\cite{SM10}.

\subsection{Our contribution and techniques}
We provide an analysis of both Pareto optimal and popular matchings in the setting of the house allocation problem with lower and upper quotas and derive tractability results for both classic and parameterized complexity. Due to space restrictions, we only sketch the main idea of most proofs in the body of the paper, and provide the full proof in the appendix.

Table~\ref{ta:results} displays a comprehensive overview of our results for classic complexity. We answer both open questions of Cechl\'arov\'a and Fleiner~\cite{CF17} and show that a Pareto optimal matching can be verified and a perfect Pareto optimal matching can be found in polynomial time if the maximum lower quota is~$2$. Furthermore, our work initiates the study of the popular house allocation problem with lower quotas by showing that even if the maximum lower quota is~$2$ it is \NP-hard to find a popular matching. However, we also present a polynomial time algorithm to verify if a given matching is popular if no lower quota exceeds~$2$, while the same problem is shown to be \NP-hard for maximum lower quota~$3$. We then reduce all three problems to the maximum weight matching problem of Arulselvan et al.~\cite{ACGMM18}, for which we firstly observe a simple reduction to the general factor problem introduced by Dudycz and Paluch~\cite{DP18}, and secondly design a faster algorithm for our special cases by combining results from \cite{DP18} and gadget techniques established by Cornu\'{e}jols~\cite{Cor88}.%

We then identify tractable sub-cases via the power of parameterized complexity, as demonstrated by Table~\ref{ta:params}. Here we again use the connection to maximum weight matchings and show how to use a treewidth-based algorithm of Arulselvan et al.~\cite{ACGMM18} to get fixed parameter tractability when parameterized by the number of applicants. Further we give a flow-based algorithm to prove fixed parameter tractability when parameterized by $\mquota$, the number of projects with a lower quota greater than~$1$. Since these two algorithms are for the maximum weight matching problem, they also apply to a recently introduced model in the area of multi-robot task allocation by Aziz et al.~\cite{ACLRW21}. Finally, by a reduction to the parametric integer programming problem~\cite{ES08}, we also show that the problem of finding a popular matching is fixed parameter tractable when parameterized by the number of projects.

\begin{table}[tb]

	\begin{center}
			\resizebox{0.7\columnwidth}{!}{ 
		\begin{tabular}{l|c|c|c}
		    \noalign{\hrule height .3mm}
			&  $\lmax \leq 2$&  $\lmax \leq 3$ & $\hopen$  \\ \noalign{\hrule height .3mm}
			\popular & \NP-c. Thm.~\ref{thrm:pop_hard_2}  & \coNP-h. Thm.~\ref{thrm:popv_lq3} & \multirow{2}{*}{\NP-h. Cor.~\ref{corr:pop_hopen}}  \\ \cline{1-3}
			\maxpareto & \multirow{3}{*}{\P\; Cor.~\ref{cor:max_size_max_match}} & \multirow{2}{*}{\NP-c. Thm.~\ref{thrm:popv_lq3}}  &  \\ \cline{1-1}\cline{4-4}
			\popver &    &   & \multirow{2}{*}{\P \; Thm.~\ref{thrm:maxw_hopen}}\\ \cline{1-1}\cline{3-3}
			
			\parver & & \NP-c.~\cite{CF17} 
			&  \\ \noalign{\hrule height .3mm}
			
		\end{tabular}
		}	
		\caption{Overview of our results in classic complexity. The four problems studied are finding a popular matching, finding a perfect Pareto optimal matching, verifying popularity, and verifying Pareto optimality.
		The columns $\lmax \leq 2$ and $\lmax \leq 3$ indicate the cases where the maximum lower quota of any project is~$2$ or $3$, respectively. The column $\hopen$ indicates the complexity of deciding whether there is matching of our desired type that opens exactly the projects in the set $\hopen$.}
		\label{ta:results}
	\end{center}
\end{table}%

\begin{table}[tb]
	\begin{center}
			\resizebox{\columnwidth}{!}{ 
		\begin{tabular}{l|c|c|c|c|c} \noalign{\hrule height .3mm}
			&  $n$ & $m$ &
			$m_{quota}$ &
			$m_{open}$ & $m_{\text{closed}}$  \\ \noalign{\hrule height .3mm}
			\popular & \W[1]-h. Thm.~\ref{thrm:pop_n} & \FPT \; Thm.~\ref{thrm:pop_fptm} & ?  & \multirow{2}{*}{\coNP-h. Thm.~\ref{thrm:pop_mopen}} &  \multirow{4}{*}{\W[1]-h. Thm.~\ref{thrm:mapom_closed}} \\ \cline{1-4}
			
			\maxpareto &  \multirow{3}{*}{\FPT \; Cor.~\ref{corr:maxw_n}}  & \multicolumn{2}{c|}{\multirow{3}{*}{\FPT \; Thm.~\ref{prop:maxwquota}}}    &  & \\ \cline{1-1} \cline{5-5}
			
			\popver &   &   \multicolumn{2}{c|}{}
			       &  \multirow{2}{*}{\W[1]-h. Thm.~\ref{thrm:w1_par_mopen}} & \\ \cline{1-1}
			
			\parver & & \multicolumn{2}{c|}{}   &  &  \\ \noalign{\hrule height .3mm}
		\end{tabular}
		}
		\caption{Overview of our parameterized results. The columns are the parameters we use in the respective cases. The first parameter $n$ is the number of applicants, $m$ is the number of projects, and $\mquota$ is the number of projects with a lower quota greater than $1$. The parameter $\mopen$ asks for a matching that opens exactly $\mopen$ projects, while $\mclosed$ asks for a matching that closes exactly $\mclosed$ projects. }
		\label{ta:params}
	\end{center}
\end{table}%

\section{Preliminaries}
\label{sec:prelims}
In this section we formally introduce our notation and the problems we consider. 
 
\subsubsection{Input graph}
We are given a set $A$ of $n$ applicants, a set $P$ of $m$ projects, and a bipartite graph $G = (A \dot{\cup }P, E)$, with $A$ and $P$ being the two sides of the bipartition. The \emph{neighborhood} $N_v$ of a vertex $v \in A \cup P$ is defined as the set of vertices $v$ is adjacent to in~$G$, and the \emph{degree} $\deg_v$ of a vertex $v$ equals~$|N_v|$. We define $\Delta_A$ to be the maximum degree of any vertex in $A$ and $\Delta_P$ to be the maximum degree of any vertex in $P$. In our model, each project $p \in P$ is equipped with a \emph{lower quota} $\ell_p \in \mathbb{N}$ and an \emph{upper quota} $u_p \in \mathbb{N}$. We refer to the maximum lower quota among all projects as $\ell_{\max}$ and the maximum upper quota as $u_{\max}$.

\subsubsection{Matching}
A \emph{matching} $M \subseteq E$ is a set of edges so that each applicant $a \in A$ is incident to at most one edge in $M$, while each project $p \in P$ is either incident to no edge in $M$ or its degree in the graph $(A \dot{\cup }P, M)$ is at least $\ell_p$ and at most~$u_p$. If applicant $a$ is assigned to project $p$ in $M$, then we write $M(a) = p$ and $a \in M(p)$. 
A matching $M$ assigns each applicant $a$ a project in $N_a$ or $a$ itself. The notation $M(a) = a$ serves convenience and it expresses that the applicant $a$ is unmatched. Conversely, the quota requirement for the projects can be expressed as $\ell_p \le \lvert M(p) \rvert \le u_p$ or $ \lvert M(p) \rvert = 0$ for each project $p \in P$. We call a project $p \in P$ with $\lvert M(p) \rvert = 0$ \emph{closed} and a project $p$ with $\ell_p \le \lvert M(p) \rvert \le u_p$ \emph{open}.
A matching $M$ is \emph{perfect} if $M(a) \neq a$ for all $a \in A$, \ie all applicants are matched to a project in~$M$. %

\subsubsection{Preferences}
Each applicant $a \in A$ has a strict order $\succ_a$ over $N_a \cup \lbrace a \rbrace$, which we call the \emph{preference list} of $a$.  For each $a \in A$ and $p \in  N_a$ we assume that $p \succ_a a$, which translates into applicant $a$ listing only the projects that are more preferable to her than staying unmatched. If the applicant is clear from the context, we simply write $\succ$ for her preference list.

We now introduce our two optimality concepts based on applicants' preferences. Given a matching $M$, we say that matching $\Mp$ \emph{dominates} $M$ if there is no $a \in A$ with $M(a) \succ_a \Mp(a)$ and there is an $a \in A$ with $\Mp(a) \succ_a M(a)$. 
We call matching $M$ \emph{Pareto optimal} if there is no matching that dominates $M$. 
Adding on to this, we call matching $\Mp$ \emph{more popular} than matching $M$ if ${\lvert\lbrace a \in A \mid \Mp(a) \succ_a M(a) \rbrace \rvert > \lvert\lbrace a \in A \mid M(a) \succ_a \Mp(a) \rbrace \rvert}$. 
Matching $M$ is \emph{popular} if there is no other matching that is more popular than~$M$.%

We are now ready to define the four problems we tackle in this paper. The first one of these is the standard popular matching problem. For two examples regarding popularity we refer the reader to Section~\ref{subs:examples} in the appendix.
\dprob{Popular house allocation with lower and upper quotas (\popularnsp)}{Graph $G = (A \dot{\cup} P, E)$, preferences $(\succ_a)_{a \in A}$, and quotas
$\ell, u\colon P \to \mathbb N$.} {Does $G$ have a popular matching?}%
Besides this we also study the complexity of verifying whether a given matching is popular.

\dprob{Popularity verification in house allocation with lower and upper quotas (\popvernsp)}{Graph $G = (A \dot{\cup} P, E)$, preferences $(\succ_a)_{a \in A}$, quotas $\ell, u\colon P \to \mathbb N$, and matching $M$.} {Does $G$ have a matching $\Mp$ that is more popular than $M$?}
Thirdly we study the problem of finding a Pareto optimal matching covering all applicants. We remind the reader that (non-perfect) Pareto optimal matchings can be found in polynomial time using a variant of the serial dictatorship method~\cite{K13, MT13, CF17}. 

\dprob{Perfect Pareto optimal house allocation with lower and upper quotas (\maxparetonsp)}{Graph $G = (A \dot{\cup} P, E)$, preferences $(\succ_a)_{a \in A}$, and quotas $\ell, u\colon P \to \mathbb N$.} {Does $G$ have a Pareto optimal matching that matches all applicants in $A$?}
Finally we also study the verification version of Pareto optimality. 

\dprob{Pareto optimality verification in house allocation with lower and upper quotas (\parvernsp)}{Graph $G = (A \dot{\cup} P, E)$, preferences $(\succ_a)_{a \in A}$, quotas $\ell, u\colon P \to \mathbb N$, and matching $M$.} {Does $G$ have a matching $\Mp$ that dominates $M$?}

\section{Connection to weighted matchings}
\label{subs:weighted_match}
Before diving into the main theorems of our paper, we present three auxiliary lemmas in this section. These allow us to reduce our problems \popvernsp, \maxparetonsp, and \parver to the following weighted many-to-one matching problem, defined by Arulselvan et al.~\cite{ACGMM18}. 
\dprob{Weighted bipartite matching with lower and upper quotas (\wmatchnsp)}{Graph $G = (A \dot{\cup} P, E)$ with quotas $\ell,u \colon P \to \mathbb{N}$, weight function $w\colon E \to \mathbb{R}$, and a bound $W \in \mathbb{R}$.} {Is there a matching $M$ with $ \sum_{e \in M}w(e) \ge W$?} 
First, we present a lemma that allows us to reduce this problem to \wmatchnsp. 
\begin{restatable}{lem}{parweight}
	For each \maxpareto instance $\mathcal{I}$, there is a \wmatch instance $\mathcal{I}'$ on the same graph, such that each maximum weight matching in $\mathcal{I}'$ corresponds to a Pareto optimal matching in~$\mathcal{I}$, with a maximum number of matched applicants. 
	The instance $\mathcal{I}'$ can be computed in polynomial time from~$\mathcal{I}$.
	\label{lem:par_weight}
\end{restatable}
\begin{proof}
	Given any $a \in A$ with preference list $p_1, \succ_a,\dots, \succ_a p_k$, we define the weight of the edge between $a$ and any $p_i$ for $i = 1, \dots, k$ to be $(k-i) + mn$.
	Let $M$ be a maximum weight matching in this new instance. 
	First $M$ has to be a maximum matching, since any larger matching would lead to a vertex being matched that was previously unmatched and thus increasing the weight of the matching by at least $mn - (n-1)m > 0$. Furthermore the matching has to be Pareto optimal, since any matching dominating it would obviously lead to a matching of larger weight. 
\qed\end{proof}

Lemma~\ref{lem:par_weight} shows that in order to check if a perfect Pareto optimal matching exists in $\mathcal{I}$, it is sufficient to find a maximum weight matching in $\mathcal{I}'$ and check if it is perfect. However, we remark that Lemma~\ref{lem:par_weight} does not hold in the reverse direction: not all perfect Pareto optimal matchings of $\mathcal{I}$ translate into a maximum weight matching in~$\mathcal{I}'$.

For \popver and \parvernsp, similar statements as first observed by Bir\'o et al.~\cite{BIM10} for one-to-one popular matchings. %
Our results state that a matching is popular / Pareto optimal if and only if it is a maximum weight matching in a certain weighted graph.
\begin{restatable}{lem}{popvermaxm}
	For each \popver/\parver instance $\mathcal{I}$ with matching $M$ there is a \wmatch instance $\mathcal{I}'$ on the same graph, such that a matching is more popular than $M$ / dominates $M$ in $\mathcal{I}$ if and only if it has a larger weight than $M$ in $\mathcal{I}'$.
	The instance $\mathcal{I}'$ can be computed in polynomial time from~$\mathcal{I}$. \label{lem:pop_ver_maxm}
\end{restatable}
\begin{proofsketch}
    As a sketch, we show the construction for the \popver case. The exact calculations and the \parver case are in the appendix. We start by defining a modified $\vote$ function. 
    For applicant $a$ and projects $p_1, p_2 \in N_a$ let \[\vote_a(p_1, p_2) = 
    \begin{cases*}2, &if $p_1 \succ_a p_2$, \\
    1, &if  $p_1 = p_2$, \\
    0, &if  $p_2 \succ_a p_1$.
    \end{cases*}\] 
    Note that based on the definition of popularity, matching $\Mp$ is more popular than matching $M$ if $\lvert\lbrace a \in A \mid \Mp(a) \succ_a M(a) \rbrace \rvert > \lvert\lbrace a \in A \mid M(a) \succ_a \Mp(a) \rbrace \rvert$, which is equivalent to $\sum_{a~\in~A} \vote_a(\Mp(a), M(a)) > n$. 
    Further, let $U(M) \coloneqq \lbrace a \in A \mid M(a) = a\rbrace$ be the set of applicants left unmatched by~$M$.
    
    For our weighted matching instance $\mathcal{I}'$ we now take the same graph as in the \popver instance $\mathcal{I}$, and introduce the weight function $w\colon E \to \lbrace 0,1,2\rbrace$
    such that for any $a \in A \setminus U(M)$ and project $p \in  N_a$ we set $w(\lbrace a,p \rbrace) = \vote_a(p, M(a))$, further for any $a \in U(M)$ and project $p \in  N_a$ we set $w(\lbrace a,p \rbrace) = 1$.
    We claim that a matching $M'$ is more popular than $M$ in $\mathcal{I}$ if and only if 
    $w(M') > n - \lvert U(M) \rvert$ in~$\mathcal{I}'$. Since the weight of $M$ in $\mathcal{I}'$ is exactly $n-\lvert U(M)\rvert$, this is sufficient to show. The intuition behind this is how agents contribute to $w(M)$ and~$w(M')$.
    \begin{itemize}
        \item each $a \in A$ with $M'(a) \succ_a M(a)$ adds 1 to $w(M')-w(M)$
        \item each $a \in A$ with $M(a) \succ_a M'(a)$ subtracts 1 from $w(M')-w(M)$
        \item each $a \in A$ with $M'(a) = M(a)$ contributes the same to $w(M)$ and $w(M)$
    \end{itemize}
    Thus $w(M') > w(M)$ is achieved if and only if $M'$ is more popular than~$M$.%
\end{proofsketch}%

\section{Constant lower quotas}
\label{sec:constantlq}
In this section, we show that all four of our problems become intractable %
when the maximum lower quota is $3$ and the maximum degree is constant. %
Furthermore, \popular is \NP-complete even if the maximum lower quota is $2$. We contrast this result by showing that the other three problems become polynomial-time solvable if the maximum lower quota is~$2$. 

\subsection{Lower quota $3$}
\label{subs:lq3}
We begin by showing that all our problems are hard for maximum lower quota~$3$ by modifying a reduction from exact cover by 3-sets by Cechl\'arov\'a and Fleiner~\cite[Theorem~6]{CF17} for their \NP-hardness proof of \parvernsp. 
\dprob{\textsc{Exact Cover by $3$-Sets (x3c)}}{Set $X = \lbrace x_1, \dots, x_{3m} \rbrace$ and set-system  $\mathcal T = \lbrace T_1, \dots, T_n \rbrace \subseteq 2^X$ such that $\lvert T_i \rvert = 3$ for all $i \in[n]$.}{Is there a subset $T^\prime\subseteq \mathcal T$ such that $T^\prime$ is a partition of $X$?}
\begin{restatable}{thm}{popvlqthr}
The problems \popver and \maxpareto are $\NP$-complete, while \popular is \coNP-hard when $\lmax =  3 = \umax $.
	\label{thrm:popv_lq3} 
\end{restatable}

\subsection{Lower quota $2$}
Next we complement the results of Section~\ref{subs:lq3}  %
and show that all problems except for \popular become polynomial-time solvable, while \popular remains \NP-complete for lower quota~$2$.
For this we use the general factor problem and a recent result from the world of factor theory. 
\optprob{\textsc{general factor problem}}{Graph $G = (V,E)$ with edge weights $w \colon E \to \mathbb{R}$ and demands $B_v \subseteq \lbrace 0, \dots, \lvert V \rvert\rbrace$ for each $v \in V$.}{Find a subgraph $H$ of maximum weight such that $\deg_H(v)\in B_v$, for each $v \in V$, where $\deg_H(v)$ is the degree of the vertex $v$ in $H$.} %
In a recent paper by Dudycz and Paluch~\cite{DP18}, a pseudo-polynomial algorithm was given for the restricted case where the maximum gap, \ie the maximum number of missing adjacent values in any degree list $B_v$, is at most~$1$.

\begin{proposition}[Dudycz and Paluch~\cite{DP18}]
		If there is no $v \in V$ and $p \ge 2$ such that $k \in B_v$, $k+1, \dots, k+p \notin B_v$ and $k+p+1 \in B_v$, then the \textsc{general factor problem} can be solved in $\mathcal O(W\lvert E \rvert \, \lvert V \rvert ^6)$ time, where $W$ is the maximum edge weight.  
	\label{thrm:dp18}  
\end{proposition}%
This theorem leads to the following corollary.
\begin{restatable}{cor}{polytime}
		Given a \wmatch instance with $\lmax = 2$, a maximum weight matching can be computed in polynomial time if the highest edge weight is polynomial in the size of the graph.   
	\label{cor:max_size_max_match}
\end{restatable} 
\begin{proof}
    This immediately follows by setting $B_a = \lbrace 0,1\rbrace$ for each applicant $a \in A$ and $B_p = \lbrace 0,\ell_p, \dots, u_p\rbrace$ \footnote{Note that in the paper of Cechl\'arov\'a and Fleiner~\cite{CF17}, the applicants technically had capacities as well, which can be easily implemented by setting $B_a = \lbrace 0, \dots, c(a) \rbrace$ for each applicant $a$ with capacity~$c(a)$.} 
    for each project $p \in P$. Since $\ell_p \leq 2$, the gaps are of size at most~$1$.
\qed\end{proof}

Next we show how to combine the methods of Cornuéjols~\cite{Cor88} and Dudycz and Paluch~\cite{DP18} to design a faster algorithm for our special case, which is also easier to implement by giving a Turing reduction to the maximum weight matching problem in graphs without quotas.
This reduction is achieved by constructing gadgets similar to the one designed by Cornuéjols~\cite{Cor88} and by 
exploiting a key lemma of Dudycz and Paluch~\cite{DP18} on the existence of augmenting paths and the structure of larger weight matchings. 

\begin{restatable}{thm}{fastalg}
	\popver and \parver can be solved in $\mathcal{O}(\lvert V \rvert^3\lvert E\rvert)$ time, while \maxpareto can be solved in $\mathcal{O}(\lvert V \rvert^3\lvert E\rvert^2)$ time if $\lmax = 2$.
	\label{thm:fastalg}
\end{restatable}

Following this, we show that it is $\NP$-complete to determine whether a popular matching exists if $\lmax = 2 = \umax$. For this we use a simple graph transformation with which popular matchings in non-bipartite instances translate to popular matchings with lower quota~$2$.
\begin{restatable}{thm}{popnpc}
	\popular is $\NP$-complete even if $\lmax = 2 = \umax$ and $\degH = 2$. 
	\label{thrm:pop_hard_2}
\end{restatable}

\section{Parameterized complexity}
\label{sec:parcomp}
In this section we study the parameterized complexity of our four problems with regard to five different parameters. These parameters are identical to the ones used by Boehmer and Heeger~\cite{BH20}: $n$, the number of applicants, $m$, the number of projects, $\mquota$, the number of projects with a lower quota greater than $1$, $\mclosed$, the number of closed projects in the matching, and finally, $m_{\text{open}}$, the number of open projects in the matching. In real life instances, one would expect the number of projects being allocated to be small in comparison to the number of students. Depending on the instance, either $m_{\text{open}}$ or $\mclosed$ would seem like a very suitable candidate for a parameterized algorithm. On the other hand, while $n$ is often large, our $\wone$-hardness result for this parameter for \popular eliminates the possibility of any fixed parameter algorithm for any smaller, maybe more realistic, sub-parameter of~$n$. %
For a brief introduction to parameterized complexity we refer to Section~\ref{subs:parameterized_comp} in the appendix. 
\subsection{Parameterization by $n$}
First we observe the following theorem from Arulselvan et al.~\cite[Theorem 4]{ACGMM18}, which was later improved by \citet{MSS21}. %
\begin{proposition}[Arulselvan et al.~\cite{ACGMM18}]
	In an instance $\mathcal{I}$ of \wmatch such that the underlying graph has a treewidth $\mathbf{tw}$, a maximum weight matching can be found in $\FPT$ time in $\mathbf{tw} + \umax$. 
\end{proposition}
From this proposition, the inequalities $\mathbf{tw} \le \min(n,m)$ and $\umax \le n$, and our results in Section~\ref{subs:weighted_match} follows the fixed parameter tractability with regard to the parameter~$n$.
\begin{restatable}{cor}{fptn}
\wmatchnsp, \popvernsp, \parvernsp, and \maxpareto are in $\FPT$ when parameterized by $n$. 
	\label{corr:maxw_n}
\label{prop:maxw_quota}
\end{restatable}
While this result relies on the machinery of tree decompositions and treewidth, we further complement this result by showing that the maximum weight matching admits a kernelization if the weights are encoded in unary. For this we exploit that if a subset of applicants could be matched to more than $n$ different projects achieving the same weight, then we can delete at least one of these projects and not change the maximum weight we can reach in the instance. 

\begin{restatable}{obs}{kerneln}
An instance of \wmatch with maximum weight $W$ admits a kernelization with  $\mathcal{O}(2^n n^2 W)$ applicants and projects.
\end{restatable}
Even though it follows from the above observation that the other three considered problems become fixed parameter tractable when parameterized by $n$, \popular remains $\wone$-hard. We show this by modifying of the proof of Boehmer and Heeger~\cite[Theorem~2]{BH20}, who reduce from the classic \textsc{multicolored independent set} problem to prove that finding a stable matching is W[1]-hard when parameterized by~$n$.
\begin{restatable}{thm}{popnwone}
	\popular parameterized by $n$ is $\wone$-hard.
    \label{thrm:pop_n}
\end{restatable}
\subsection{Parameterization by $m$}
Next we turn to \popular and show that the problem becomes fixed parameter tractable when parameterized by the number of projects $m$. For this we show how to reduce \popular to the \textsc{parametric integer program} problem. 
\dprob{\textsc{parametric integer program}}{Matrices $B \in \mathbb{Q}^{n \times m}, C \in \mathbb{R}^{m \times k}$, and a vector $d \in \mathbb{R}^k$.}{For any $b \in \mathbb{Z}^m$, such that $Cb \le d$, does there exist an $x \in \mathbb{Z}^n$ such that $Bx \le b$?} 
The feasibility 
of a \textsc{parametric integer program} can be decided in $\mathcal O(f(n,m) poly(\lvert\lvert B,C,d\rvert \rvert_\infty, k))$ time, as shown by Eisenbrand and Shmonin~\cite{ES08}. 
Using this we can now show that \popular is indeed in FPT when parameterized by $m$. We first construct a matrix $C$ and vector $d$ such that all feasible matchings are represented by solutions to $Cb \le d$. Then we construct the matrix $B$ such that any solution to $Bx \le b$ represents a feasible matching that is more popular than the matching represented by $b$, thus ensuring that the parametric integer program is feasible if and only if no popular matching exists.  %
To bound the dimensions of the matrices we use that we can assign each agent a type based on their preferences, of which there can be at most $\mathcal{O}((m+1)!)$. 

\begin{restatable}{thm}{popfptm}
\popular is in $\fpt$ when parameterized by $m$.
    \label{thrm:pop_fptm}
\end{restatable}
\begin{proofsketch}
Our goal is to encode the matching instance in such a way that the resulting \textsc{parametric integer program} is feasible if and only if no popular matching exists. For this we choose the matrix $B$ and vector $d$ so that all possible matchings can be represented by a vector $b$ satisfying $Cb \le d$. Further the matrix $B$ should be chosen in such a way that the vector $x$ represents a matching that is more popular than the matching represented by $b$. 
    
    \subsubsection{Notation}
    Since we parameterize by $m$, each applicant $a \in A$ can be uniquely identified by her preference structure over $P$. There are at most $\mathcal{O}((m+1)!)$ different preference structures, hence 
    we can partition $A$ into $t \in \mathcal{O}((m+1)!)$ types. Let $A_1, \dots, A_t$ be this partition such that any two applicants in the same set have identical preference lists. We refer to the projects that appear in the preference lists of applicants in $A_i$ as~$N_i$.
    Furthermore we slightly alter the notation of the previous sections to follow~\cite{BIM10} and define the vote of a vertex as 
    \[\vote_a(p_1, p_2) = 
    \begin{cases*}1, &if $p_1 \succ_a p_2$, \\
    0, &if  $p_1 = p_2$, \\
    -1, &if  $p_2 \succ_a p_1$.
    \end{cases*} \] By the definition of popularity, matching $\Mp$ is more popular than $M$ if and only if $\sum_{a \in A} \vote_a(\Mp(a), M(a)) \ge 1$ holds. For easier notation for any type $i \in  [t]$ we define $\vote_i(p_1, p_2) =\vote_a(p_1, p_2)$ where $a \in A_i$ is an applicant of type $i$. 
    
    \subsubsection{Construction of $C$}
    As noted earlier, our goal is to construct the linear program represented by $B$ and $d$ in such a way that every feasible matching is a solution to this linear program. 
    To reach this, for each $i \in [t]$ and $p \in N_i$ we create a variable $x_i^p$ that should indicate how many applicants of type $A_i$ are matched to project $p$. Moreover we add one variable $x_i^i$ indicating the number of unmatched applicants of type $i$. Furthermore for each project $p \in P$ we create a variable $o_p$ that should indicate whether project $p$ is open or closed. 
    This now leads us to the following linear program
\begin{alignat}{4}
        \displaystyle\sum\limits_{p\in N_i \cup \lbrace i \rbrace}   x_i^p &=& &\lvert A_i \rvert, \quad &\mbox{\ for each $i \in [t]$}\label{eq1}\\
        \displaystyle\sum\limits_{i \in [t] \colon p\in N_i} x_i^p - o_p u_p &\le& &0, \quad &\mbox{\ for each $p\in P$} \label{eq2}\\
        \displaystyle\sum\limits_{i \in [t] \colon p\in N_i} x_i^p - o_p \ell_p &\ge& &0,  \quad &\mbox{\ for each $p\in P$} \label{eq3}\\
        x_i^p &\ge& &0, \quad &\mbox{\ for each $ i \in [t]$ and $p\in N_i$}  \label{eq4}\\
        o_p &\in& &[0,1], \quad &\mbox{\ for each $p\in P$} \label{eq5}
\end{alignat}%
Here Constraint~(\ref{eq1}) enforces that all applicants are either matched or unmatched. 
With Constraints~(\ref{eq2}) and (\ref{eq3}) we ensure that the number of applicants matched to an open project is between its lower and upper quota and the number of applicants matched to a closed project is $0$.
   Note that each feasible matching is a solution to this ILP. Currently, any solution to this ILP is of the form $(x_1^{p_1}, \dots, x_t^{p_m}, o_{p_1}, \dots,  o_{p_m})$. This, however, is not enough to fully model the ILP we need for~$B$. As a first step, for each variable of the form $x_i^p$ we add a second copy. Furthermore we add $2t$ variables $b_1, \dots, b_{2t}$ with $b_{2i} = \lvert A_i\rvert = b_{2i+1}$, then we add $2m$ variables that are forced to be $0$, and one variable that is forced to be $-1$. After this, each solution $b$ is of the form $(\lvert A_1\rvert, \lvert A_1 \rvert, \dots, \lvert A_t \rvert, \lvert A_t \rvert, x_1^{p_1},x_1^{p_1}, \dots, x_t^{p_m},x_t^{p_m}, o_{p_1}, \dots,  o_{p_m}, \underbrace{0, \dots, 0}_{2m \text{ times}}, -1)$.
   
   \subsubsection{Construction of $B$}
   We design $B$ to ensure that there is a matching $M'$ that is more popular than the matching $M$ induced by~$b$. For any type $i \in [t]$ and projects $p,p' \in N_i\cup\lbrace i \rbrace$, we create a variable $x_i^{p \rightarrow p'}$ that should indicate the number of applicants of type $i$ who were matched to project $p$ in $M$ and are matched to project $p'$ in $M'$. Furthermore, for each project $p \in P$, we again create a variable $o'_p$ indicating whether project $p$ is open or closed in~$M'$.
   This now allows us to construct the final ILP.
   \begin{alignat}{4}
        \displaystyle\sum\limits_{p, p'\in N_i \cup \lbrace i \rbrace}   x_i^{p \rightarrow p'} &=& &\lvert A_i \rvert,  \quad &\mbox{ for each $i \in [t]$} \label{eq6}\\
         \displaystyle\sum\limits_{p' \in P\cup \lbrace i \rbrace}   x_i^{p \rightarrow p'} &=& &x_i^p , \quad &\mbox{ for each $i \in [t]$ and $p \in N_i$} \label{eq7}\\
        \displaystyle\sum\limits_{i \in [t], p \in P} x_i^{p \rightarrow p'} - o'_{p'} u_{p'} &\le& &0, \quad &\mbox{ for each $p' \in P$} \label{eq8} \\
       \displaystyle\sum\limits_{i \in [t], p \in P} x_i^{p \rightarrow p'} - o'_{p'} \ell_{p'} &\ge& &0, \quad &\mbox{ for each $p' \in P$} \label{eq9}\\
       \displaystyle\sum\limits_{i \in [t] p, p'\in N_i \cup \lbrace i \rbrace}  -\vote_i(p',p)x_i^{p \rightarrow p'} &\le& &-1 ,& \quad \label{eq10}\\
       x_i^{p \rightarrow p'} &\ge& &0, \quad &\mbox{ for each $i \in [t]$ and $p,p' \in N_i$} \label{eq11}\\
        o'_p &\in& &[0,1], \quad &\mbox{ for each $p \in P$} 
   \end{alignat}
   Constraint~(\ref{eq6}) ensures that each applicant has a new partner in $M'$, and Constraint~(\ref{eq7}) guarantees that each applicant matched to some $p$ in $M$ now has a new partner. Constraints~(\ref{eq8}) and (\ref{eq9}) enforce the lower and upper quota constraints, and finally Constraint~(\ref{eq10}) ensures that $M'$ is more popular than~$M$.%
\end{proofsketch}
\subsection{Parameterization by $\mquota$}%
While for $\popular$ we were able to show fixed parameter tractability with regard to $m$, we now improve this parameter for the other three problems by considering $\mquota$, \ie the number of projects with a lower quota greater than $1$.
In order to solve this problem we turn to a special subcase of our matching problems defined by Boehmer and Heeger~\cite{BH20}, namely the task of deciding whether, given a certain set of projects $\hopen$, there is a matching with our desired property that opens exactly the projects in $\hopen$.
 We show that finding a maximum weight matching 
 that opens exactly the projects in $\hopen$ is polynomial-time solvable. First we give an algorithm for simply finding a matching that dominates a given matching and afterwards we modify this algorithm to find a maximum weight matching. For this algorithm we utilize the \textsc{feasible flow with demands} problem, which can easily be reduced to the well-known maximum flow problem, as argued for instance in the textbook of Erickson~\cite{Erick19}.
\dprob{\textsc{feasible flow with demands}}{Directed graph $G = (V,E)$, capacities $C\colon E \to \mathbb N$, and demands $D \colon E \to \mathbb N$.}{Is there a flow $f \colon E \to \mathbb{N}$ such that $D(e) \le f(e) \le C(e)$ for all $e \in E$?} 
\begin{restatable}{thm}{povhopen}
	Given an instance $\mathcal{I}$ of \parver with input matching $M$ and a set $\hopen \subseteq P$, we can decide in polynomial time whether a matching $\Mp$ exists that dominates $M$ and opens exactly the projects in the set $\hopen$.
	\label{thrm:pov_hopen}
\end{restatable}
\begin{proofsketch}
We construct an instance of \textsc{feasible flow with demands} to solve the  problem. The proof of correctness for this construction is in the appendix. 
	First we assume that some $a \in A$ is given who 
	will be the designated applicant to receive a better partner in $\Mp$. 
	Now we create a directed bipartite graph 
	$G_a$ with vertices $A \cup \hopen \cup \{p_\perp,s,t\}$, where $p_\perp$ will represent the unmatched applicants.  We connect $s$ to all vertices in $A$ with demand and capacity of $1$ on the edge. 
	Next we connect $a$ to all projects in $\hopen$ that $a$ prefers to $M(a)$ and all applicants in $A \setminus \lbrace a \rbrace$ get connected to all projects in $\hopen$ which they prefer to $M(a)$ as well as to $M(a)$ if $M(a)$ is in $\hopen$, further if they are unmatched in $M$, we also connect them to $p_\perp$. 
	All of these edges have no demand and a capacity of~$1$. Finally we connect every project $p \in  \hopen$ to $t$, each with a demand of $\ell_p$ and with a capacity of $u_p$ and we connect $p_\perp$ to $t$ with no demand and infinite capacity. 
	We solve the \textsc{feasible flow with demands} problem for all $a \in A$ in $G_a$ and if some $a \in A$ exists for which the demands are satisfiable, we return the respective induced matching as $\Mp$. That is, if a flow of $1$ goes from $a$ to $p$, we match $a$ to $p$ in $\Mp$ and if a flow of $1$ goes from $a$ to $p_\perp$, we leave $a$ unmatched in $\Mp$.
	Otherwise we return that our matching is Pareto optimal.%
\end{proofsketch}
\begin{figure}[htb]
\begin{minipage}{0.6 \textwidth}    \includegraphics[scale = 0.85]{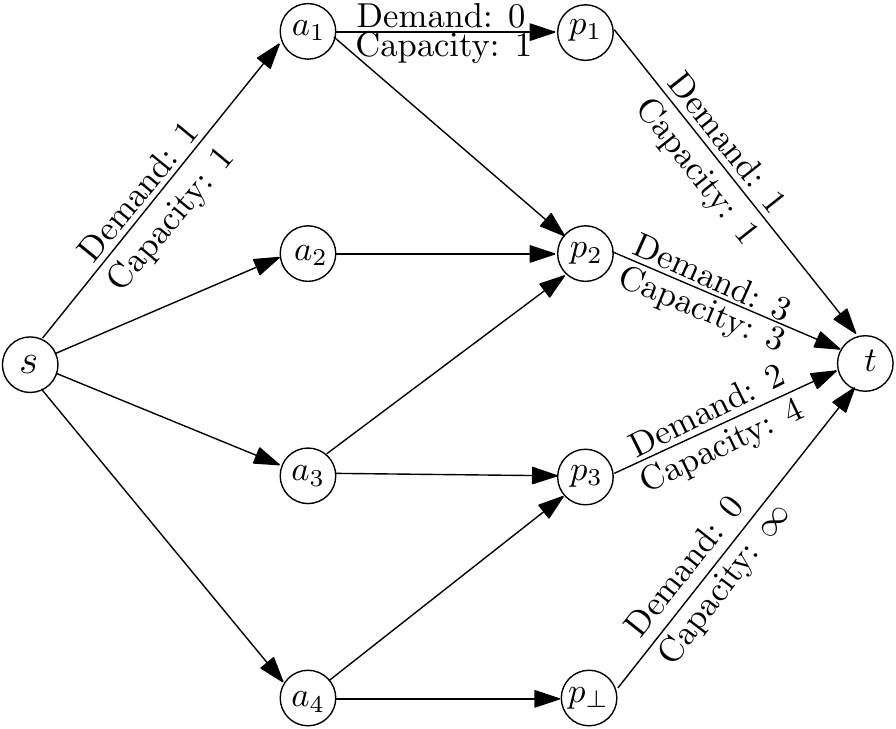}
\end{minipage}\hfill
\begin{minipage}{0.3 \textwidth}
    \begin{tabular}{rll}
       $\phantom{1111}\ell_{p_1} =$& 1 = &$u_{p_1}$\\
       $\ell_{p_2} =$& 3 = &$u_{p_2}$\\
       $\ell_{p_3} =$& 2, &$u_{p_3} = 4$\\
       $\ell_{p_4} =$& 2 = &$u_{p_4}$
    \end{tabular}
    \begin{align*}
       a_1 \colon& p_2 \succ p_1 \succ p_3 \\
       a_2 \colon& p_2 \succ p_4\\
       a_3 \colon& p_3 \succ p_2 \succ p_4 \\
       a_4 \colon& p_3
    \end{align*}
\end{minipage}
    \caption{Consider an instance of \parver with four applicants $a_1, \dots, a_4$ and four projects $p_1, \dots, p_4$ 
    with quotas and preference lists as shown on the right side of the figure.
    If we consider the matching $M(a_1) = p_1, M(a_2) = p_4, M(a_3) = p_4$, and $M(a_4) = a_4$ with $\hopen = \lbrace p_1, p_2, p_3 \rbrace$, the above instance is the \textsc{feasible flow with demands} instance created in Theorem~\ref{thrm:pov_hopen}.} 
\end{figure} %
Next we show how to generalize this idea to prove that finding a maximum weight matching that opens exactly the projects in $\hopen$ can be done in polynomial time. For this we need the slightly more complicated but still polynomial-time solvable \textsc{max-cost circulation} problem, see Tardos~\cite{T85}. 
\optprob{\textsc{max-cost circulation}}{Directed graph $G = (V,E)$, costs $p \colon E \to \mathbb R$, capacities $C\colon E \to \mathbb N$, and demands $D \colon E \to \mathbb N$.}{Find a flow $f \colon E \to \mathbb N$ of maximum cost such that for any $v \in V$ it holds that $\sum_{(v,w) \in E} f((v,w)) - \sum_{(w,v) \in E} f((w,v)) = 0$, \ie all flow gets conserved and such that for any edge $e \in E$ it holds that $D(e) \le f(e) \le C(e)$.}
\begin{restatable}{thm}{maxwhopen}
	Given an instance $\mathcal{I}$ of \wmatch and a set $\hopen \subseteq P$, we can find in polynomial time a maximum weight matching that opens exactly the projects in $\hopen$. 
	\label{thrm:maxw_hopen}
\end{restatable}
Using this theorem, we can calculate maximum weight matchings in FPT time in $\mquota$. To do so, we use the fact that projects with a lower quota of at most $1$ are always open and then iterate through all projects with a lower quota of at least~$2$.
\begin{restatable}{thm}{fptmquota}
\label{prop:maxwquota}
Given an instance $\mathcal{I}$ of \wmatchnsp, 
a maximum weight matching can be computed in $\mathcal O(2^{m_{\text{quota}}} poly( \lvert \mathcal{I} \rvert ))$, where $ \lvert \mathcal{I} \rvert = nm+$ the input size of the weights. 
\end{restatable}
Using the lemmas in Section~\ref{subs:weighted_match} this now implies that \popvernsp, \maxparetonsp, and \parver are all fixed parameter tractable when parameterized bz $\mquota$.
\subsection{Parameterization by $m_{\text{open}}$}%
The next parameter we investigate is $\mopen$, the number of projects that are open in the output matching. We show hardness for all our problems, by reducing from either \textsc{x3c} or the \textsc{exact unique hitting set} problem. %
\begin{restatable}{thm}{woneparmopen}
	Given an instance $\mathcal{I}$ of \parvernsp/\popver with a matching $M$ and a parameter $\mopen$, it is both $\w[1]$-hard to decide whether there is a matching that dominates / is more popular than $M$ and opens exactly $\mopen$ projects. This remains true even if $M$ opens exactly $1$ project. 
	\label{thrm:w1_par_mopen}
\end{restatable} 
This approach also generalizes to \popular and \maxpareto with the difference that here it is even hard to decide whether there is a desired matching with only one open project.
\begin{restatable}{thm}{wonepopmopen}
	Given an instance $\mathcal{I}$ of \popularnsp/\maxparetonsp, it is \coNP-hard to decide whether there is a popular / perfect Pareto optimal matching that opens exactly $1$ project. 
	\label{thrm:pop_mopen}
\end{restatable}
As a corollary of our previous proof, we obtain that given a set $\hopen \subseteq P$, it is \coNP-hard to determine whether a perfect Pareto optimal matching or a popular matching opening exactly the projects in $\hopen$ exists. 
\begin{restatable}{cor}{coNPHopen}
			Deciding whether there exists a popular matching / perfect Pareto optimal matching that opens exactly the projects in a given subset $\hopen \subseteq P$ is \coNP-hard, even if $\lvert \hopen\rvert = 1$.  \footnote{Note that this is not a contradiction to Theorem~\ref{thrm:maxw_hopen}, since given an instance of \maxparetonsp, a maximum weight matching opening only projects in $\hopen$ does not need to correspond to a perfect Pareto optimal matching and vice versa.}
	\label{corr:pop_hopen}
\end{restatable} 
\subsection{Parameterization by $\mclosed$}
As our final parameter we turn to $\mclosed$, the number of closed projects in our desired matching and show that all four problems become W[1]-hard when parameterized by $\mclosed$. For these results we reduce from the classic problems \textsc{multicolored independent set} and \textsc{multicolored clique}.
\begin{restatable}{thm}{mapomclosed}
	 Given an instance $\mathcal{I}$ of \maxparetonsp, \parvernsp, \popularnsp, or \popver and parameter $\mclosed$, it is $\wone$-hard to decide whether a matching of our desired type exists that closes exactly $\mclosed$ projects.
    \label{thrm:mapom_closed}
\end{restatable}

\section{Open questions}
There are two major open questions and future research directions that could be derived from our paper. Firstly, the question whether \popular is in FPT when parameterized by $\mquota$ is still open. Secondly, even after the papers of Arulselvan et al.~\cite{ACGMM18}, Dudycz and Paluch~\cite{DP18}, and now our paper it is still open whether \wmatch with maximum lower quota $2$ or the general factor problem with gap at most $1$ can be solved in polynomial time, \ie by eliminating the linear factor on $W$, the largest weight.

\bibliography{quotas}

\newpage
\section{Appendix}

\subsection{Parameterized complexity}
\label{subs:parameterized_comp}
Here we provide a short overview of parameterized complexity and the problems we use in our reductions. For further background on parameterized complexity and algorithms we refer the reader to the textbook by Cygan et al.~\cite{cygan2015parameterized}.

A parameterized problem consists of two parts, the instance $\mathcal{I}$ and a parameter~$k$. The problem is in FPT (and called fixed parameter tractable) if there is some algorithm solving it in time $\mathcal{O}(f(k) \lvert \mathcal{I} \rvert ^{\mathcal{O}(1)})$ for some computable function~$f$. A parameterized problem $P$ is reducible to another parameterized problem $\hat P$ if there is some function $f$ that can be computed in time $\mathcal{O}(g(k) \lvert \mathcal{I} \rvert ^{\mathcal{O}(1)})$ for some computable function $g$, such that $(\mathcal{I}, k) \in P$ if and only if $f(\mathcal{I},k)  = ( \hat{\mathcal{I}}, \hat k)\in \hat P$ and $\hat k \le h(k)$ for some computable function $h$. It is easy to see that if $\hat P$ is in FPT, then so is $P$ as well. Related to this is the concept of a kernelization. Given an instance $(\mathcal{I},k)$ of a parameterized problem $P$, a kernelization is an algorithm that computes another instance $(\mathcal{I}',k')$ of $P$ in polynomial time, such that $(\mathcal{I},k) \in P$ if and only if $(\mathcal{I}',k') \in P$ and $\mathcal{I}' + k' \le f(k)$ for some computable function $k$. This can be seen as applying provably correct data reductions. It is commonly known that a problem $P$ is in FPT if and only if it has a kernelization, see~\cite{cygan2015parameterized}.   The class W[1] is now the class of all parameterized problems that have a parameterized reduction to the clique problem, parameterized by the solution size. We call a problem W[1]-hard if any problem in W[1] can be reduced to it in FPT time, further we call a problem W[1]-complete if it is W[1]-hard and also in W[1]. Since it is widely assumed that clique is not in FPT, this assumption extends to all W[1]-hard problems. 

The W[1]-complete problems we reduce from in this paper are as follows. 
Firstly we use a highly restricted version of the hitting set problem.
\dprob{\textsc{Exact unique hitting set~\cite{cygan2015parameterized}}}{Set of elements $X$, set system $\mathcal{T} \subseteq 2^X$, parameter $k$.} {Is there a set $H \subseteq X$ with $\lvert H \rvert = k$, such that $\lvert T \cap H \rvert = 1$ for all $T \in \mathcal{T}$?}
Furthermore due to their nice structural properties we use the multicolored versions of independent set and clique, shown to be W[1]-complete by Pietrzak~\cite{P03}.

\dprob{\textsc{Multicolored clique}}{Graph $G= (V,E)$ and partition $V_1, \dots, V_k$ of $V$.} {Is there a clique $C \subseteq V$, \ie a vertex set inducing a complete graph, such that $\lvert V_c \cap C\rvert = 1$ for all $c \in [k]$?}

\dprob{\textsc{Multicolored independent set}}{Graph $G= (V,E)$ and partition $V_1, \dots, V_k$ of $V$.} {Is there an independent set $I \subseteq V$, \ie a vertex set inducing a graph without any edges, such that $\lvert V_c \cap I\rvert = 1$ for all $c \in [k]$?}

\subsection{Examples}
\label{subs:examples}
We now present two characteristic instances we will later repeatedly use in our proofs.

\begin{restatable}{obs}{condorcetone}
Consider the instance $\mathcal{I}$ of \popular with three applicants $a_1, a_2, a_3$ and three projects $p_1, p_2, p_3$, each with a lower and upper quota of $1$, such that the preference list of each applicants is $p_1 \succ p_2 \succ p_3$. This instance does not admit a popular matching. 
\label{obs:condorcet}
\end{restatable}
\begin{proof}
    A matching that is not perfect cannot be popular, since the matching obtained by additionally pairing an unmatched applicant to an unmatched project is preferred by this one applicant, while the other applicants are indifferent. 
    Since the instance is symmetric, we can assume without loss of generality that if there is a popular matching then it is ${M = \lbrace (a_1, p_1),(a_2,p_2),(a_3,p_3) \rbrace}$. 
    Let us consider the matching $\Mp = \lbrace (a_1, p_3),(a_2,p_1),(a_3,p_2) \rbrace$. 
    Since $\Mp$ is preferred to $M$ by $a_2$ and $a_3$, while only $a_1$ prefers $M$ to $\Mp$, $\Mp$ is more popular than~$M$.
\qed\end{proof}
This also implies that for any instance containing three projects with lower and upper quota $1$ and three applicants with the same relative preference order over these three projects, in any popular matching, at least one applicant needs to be matched to a project different from those three. Note that if we would delete one of the three applicants, a popular matching would indeed be possible, since we could match the remaining two applicants to $p_1$ and $p_2$, thus ensuring that none of them could improve without the other applicant getting worse. 
As the second example we introduce a small gadget we also use in a later proof, which will also be an instance of the famous Condorcet cycle.
\begin{restatable}{obs}{condorcettwo}
Consider the instance $\mathcal{I}$ of \popular with three applicants $a_1, a_2, a_3$ and three projects $p_1, p_2, p_3$, each with a lower and upper quota of $3$, such that the preference list of $a_1$ is $p_1 \succ_{a_1} p_2 \succ_{a_1} p_3$, the preference list of $a_2$ is $p_2 \succ_{a_2} p_3 \succ_{a_2} p_1$ and the preference list of $a_3$ is $p_3 \succ_{a_3} p_1 \succ_{a_3} p_2$, \ie the preference lists are just cyclically shifted between the applicants. This instance does not admit a popular matching.
\label{obs:cycl_low3}
\end{restatable}
\begin{proof}
    Any non-perfect matching has to be empty in this instance due to the lower quota of $3$ on each project. Thus matching all three applicants to any of the projects would lead to a matching that is preferred by all three applicants. Therefore we only need to consider the three perfect matchings
    \begin{itemize}
        \item $M_1$ with  $M_1(p_1) = \lbrace a_1, a_2, a_3\rbrace$,
        \item $M_2$ with  $M_2(p_2) = \lbrace a_1, a_2, a_3\rbrace$,
        \item $M_3$ with  $M_3(p_3) = \lbrace a_1, a_2, a_3\rbrace$.
    \end{itemize}  
    Then the following statements hold.
    \begin{itemize}
        \item The applicants $a_2$ and $a_3$ prefer $M_3$ to $M_1$, while $a_1$ prefers $M_1$ to $M_3$, making $M_3$ more popular than $M_1$.
        \item  The applicants $a_1$ and $a_2$ prefer $M_2$ to $M_3$, while $a_3$ prefers $M_3$ to $M_2$, making $M_2$ more popular than $M_3$.
        \item The applicants $a_1$ and $a_3$ prefer $M_1$ to $M_2$, while $a_2$ prefers $M_2$ to $M_1$, making $M_1$ more popular than $M_2$.
        \end{itemize}
     Thus this instance admits no popular matching.
\qed\end{proof}

\subsection{Missing proof %
from Section~\ref{subs:weighted_match}}

\popvermaxm*
\begin{proof}
    We start with the \popver case. 
    For applicant $a$ and projects $p_1, p_2 \in N_a$ let \[\vote_a(p_1, p_2) = 
    \begin{cases*}2, &if $p_1 \succ_a p_2$, \\
    1, &if  $p_1 = p_2$, \\
    0, &if  $p_2 \succ_a p_1$.
    \end{cases*}\] 
    Note that based on the definition of popularity, matching $\Mp$ is more popular than matching $M$ if $\lvert\lbrace a \in A \mid \Mp(a) \succ_a M(a) \rbrace \rvert > \lvert\lbrace a \in A \mid M(a) \succ_a \Mp(a) \rbrace \rvert$ 
    , which is equivalent to $\sum_{a \in A} \vote_a(\Mp(a), M(a)) > n$. 
    Further let $U(M) \coloneqq \lbrace a \in A \mid M(a) = a\rbrace$ be the set of applicants left unmatched by~$M$.
    
    \subsubsection{Construction}
    For our weighted matching instance $\mathcal{I}'$ we now take the same graph as in the original instance $\mathcal{I}$, and introduce the weight function $w\colon E \to \lbrace 0,1,2\rbrace$
    such that for any $a \in A \setminus U(M)$ and project $p \in  N_a$ we set $w(\lbrace a,p \rbrace) = \vote_a(p, M(a))$, further for any $a \in U(M)$ and project $p \in  N_a$ we set $w(\lbrace a,p \rbrace) = 1$.
    We claim that a matching $M'$ is more popular than $M$ in $\mathcal{I}$ if and only if 
    $w(M') > n - \lvert U(M) \rvert$ in~$\mathcal{I}'$. Since the weight of $M$ in $\mathcal{I}'$ is exactly $n-\lvert U(M)\rvert$ this is sufficient to show. 
    
    \subsubsection{Correctness}
    Let $\Mp$ be a matching that is more popular than $M$ in~$\mathcal{I}$. By the definition of popularity it holds that $\lvert\lbrace a \in A \mid \Mp(a) \succ_a M(a) \rbrace \rvert > \lvert\lbrace a \in A \mid M(a) \succ_a \Mp(a) \rbrace \rvert$. We now calculate the weight of $M'$ in~$\mathcal{I}'$. 
    \begin{align*} w(\Mp) &= 
    2 \lvert\lbrace a \in A \setminus U(M) \mid \Mp(a) \succ_a M(a) \rbrace \rvert  + \lvert\lbrace a \in U(M) \mid \Mp(a) \succ_a a \rbrace \rvert \\ 
    &+ \lvert\lbrace a \in A \setminus U(M) \mid \Mp(a) = M(a) \rbrace \rvert  \\ 
    & \stackrel{(1)}{>}\lvert\lbrace a \in A \setminus U(M) \mid \Mp(a) \succ_a M(a) \rbrace \rvert  +  \lvert\lbrace a \in A \mid M(a) \succ_a \Mp(a) \rbrace \rvert \\ 
    &+ \lvert\lbrace a \in A \setminus U(M) \mid \Mp(a) = M(a) \rbrace \rvert \\
    &\ge  \lvert\lbrace a \in A \setminus U(M) \mid \Mp(a) \succ_a M(a) \rbrace \rvert  +  \lvert\lbrace a \in A \setminus U(M) \mid M(a) \succ_a \Mp(a) \rbrace \rvert \\ 
    &+ \lvert\lbrace a \in A \setminus U(M) \mid \Mp(a) = M(a) \rbrace \rvert \\
    & = \lvert A \setminus U(M) \rvert \\
    &= n - \lvert U(M) \rvert \end{align*}
    Inequality~(1) follows from the fact that $\lvert\lbrace a \in A \setminus U(M) \mid \Mp(a) \succ_a M(a) \rbrace \rvert  + \lvert\lbrace a \in U(M) \mid \Mp(a) \succ_a a \rbrace \rvert = \lvert\lbrace a \in A \mid \Mp(a) \succ_a M(a) \rbrace \rvert > \lvert\lbrace a \in A \mid M(a) \succ_a \Mp(a) \rbrace \rvert$. 
    
    Now assume that $\Mp$ is not more popular than $M$ in~$\mathcal{I}$. Then by definition we know that $\lvert\lbrace a \in A \mid \Mp(a) \succ_a M(a) \rbrace \rvert \le  \lvert\lbrace a \in A \mid M(a) \succ_a \Mp(a) \rbrace \rvert$. By analogous calculations as in the previous case we get the following. 
    \begin{align*}
        w(\Mp) &= 2 \lvert\lbrace a \in A \setminus U(M) \mid \Mp(a) \succ_a M(a) \rbrace \rvert  + \lvert\lbrace a \in U(M) \mid \Mp(a) \succ_a a \rbrace \rvert  \\
        &+ \lvert\lbrace a \in A \setminus U(M) \mid \Mp(a)= M(a) \rbrace \rvert \\
        &\stackrel{(1)}{=} n + \lvert\lbrace a \in A \setminus U(M) \mid \Mp(a) \succ_a M(a) \rbrace \rvert -  \lvert\lbrace a \in U(M) \mid \Mp(a) = a \rbrace \rvert \\
        & -  \lvert\lbrace a \in A \mid M(a) \succ_a \Mp(a) \rbrace \rvert  \\
        &\stackrel{(2)}{=} n + \lvert\lbrace a \in A \setminus U(M) \mid \Mp(a) \succ_a M(a) \rbrace \rvert + \lvert\lbrace a \in U(M) \mid \Mp(a) \succ_a a \rbrace \rvert\\
        & - \lvert\lbrace a \in U(M) \mid \Mp(a) \succ_a a \rbrace \rvert -  \lvert\lbrace a \in U(M) \mid \Mp(a) = a \rbrace \rvert \\
        & -  \lvert\lbrace a \in A \mid M(a) \succ_a \Mp(a) \rbrace \rvert  \\
        & = n + \lvert\lbrace a \in A \mid \Mp(a) \succ_a M(a) \rbrace \rvert \\
        & - \lvert U(M) \rvert  - \lvert\lbrace a \in A \mid M(a) \succ_a \Mp(a) \rbrace \rvert\\ 
        &\stackrel{(3)}{\le} n + \lvert\lbrace a \in A \mid M(a) \succ_a \Mp(a) \rbrace \rvert -  \lvert U(M) \rvert \\
        & -  \lvert\lbrace a \in A \mid M(a) \succ_a \Mp(a) \rbrace \rvert  \\
        & = n - \lvert U(M) \rvert
    \end{align*} 
    Equality~$(1)$ follows from the fact that 
    \[
    n = \lvert\lbrace a \in A \mid \Mp(a) \succ_a M(a) \rbrace \rvert + \lvert\lbrace a \in A \mid M(a) \succ_a \Mp(a) \rbrace \rvert + \lvert\lbrace a \in A \mid M(a) = \Mp(a) \rbrace \rvert.
    \] 
    In Equality~(2) we added $0 =  \lvert\lbrace a \in U(M) \mid \Mp(a) \succ_a a \rbrace \rvert - \lvert\lbrace a \in U(M) \mid \Mp(a) \succ_a a \rbrace \rvert $ and in Inequality~(3) we used $\lvert\lbrace a \in A \mid \Mp(a) \succ_a M(a) \rbrace \rvert \le  \lvert\lbrace a \in A \mid M(a) \succ_a \Mp(a) \rbrace \rvert$, which follows from our assumption that $M'$ is not more popular than~$M$.
    
    With this we have shown that a matching $M$ is popular in $\mathcal{I}$ if and only if there is no matching with weight larger than $n-\lvert U(M) \rvert$ in~$\mathcal{I}'$.
    
    With regard to the \parver case in our instance $\mathcal{I}'$ we again have the same graph as in $\mathcal{I}$ and add the weight function $w\colon E\to\lbrace 0,1,n,n+1\rbrace$, such that for any $a \in A \setminus U(M)$ and $p \in  N_a$ 
    \begin{equation*}
    w(\lbrace a,p \rbrace) = 
    \begin{cases*}
    n+1, &if  $h \succ_a M(a)$ \\
    n, &if  $h = M(a)$\\
    0, &if $M(a) \succ_a h$
    \end{cases*}
    \end{equation*}
    and if $a \in U(M)$ and $p \in  N_a$ we set $w(\lbrace a,p \rbrace) = 1$. We claim that a matching $\Mp$ dominates $M$ in $\mathcal{I}$ if and only if $w(\Mp) > n(n-\lvert U(M)\rvert) = w(M).$ 
    
    \subsubsection{Correctness}
    First assume that there is a matching $\Mp$ that dominates $M$ in $\mathcal{I}$. Then there is at least one applicant who prefers $\Mp$ to $M$ while no applicant prefers $M$ to $\Mp$ and thus based on the definition, the weight of the matching $\Mp$ in $\mathcal{I}'$ is larger than $ n(n-\lvert U(M)\rvert)$.
    Similarly, if $\Mp$ does not dominate $M$ in $\mathcal{I}$, then either there is an applicant in $A \setminus U(M)$ preferring $M$ to $\Mp$, thus contributing no weight to the matching and therefore the weight of the matching $\Mp$ in $\mathcal{I}'$ cannot possibly reach $n(n-\lvert U(M)\rvert)$ or each applicant in $\mathcal{I}$ is indifferent between $M$ and $\Mp$, which implies $M = \Mp$ and thus $w(\Mp) = n(n-\lvert U(M)\rvert)$ in~$\mathcal{I}'$.
\qed\end{proof}

\subsection{Missing proofs from Section~\ref{sec:constantlq}}

\popvlqthr*

For better readability we split this theorem into three separate statements. 

\begin{restatable}{thm}{popvlqthrapp}
\popver is \NP-complete, even with $\lmax = 3 = \umax$, $\degA = 4$, and $\degH = 3$.
	\label{thrm:pop_lq3}
\end{restatable}
\begin{proof}
    The membership in \NP \; immediately follows from the fact it can be checked in polynomial time whether a matching is feasible and if it is more popular than another matching. 
	To show the hardness we reduce from exact cover by 3-sets. 
	 We note that \textsc{x3c} is \NP-complete, even if each element of $X$ appears in at most three sets~\cite{GJ79}. Furthermore, we can assume that $m$ is odd, because for an even $m$ we can simply add three elements that appear in no other set, and a set covering exactly them.
	
	\subsubsection{Construction}
	We create three sets of applicants $B = \lbrace b_1, \dots, b_{3m} \rbrace$, $Y = \lbrace y_1, \dots, y_{3m} \rbrace$, $Z = \lbrace z_1, \dots, z_{3m} \rbrace$ and we set $A = B \cup Y \cup Z$. Furthermore we create three sets of projects $C = \lbrace c_1, \dots, c_n\rbrace$, $D = \lbrace d_1, \dots, d_{3m} \rbrace$, $E = \lbrace e_1, \dots, e_{3m} \rbrace$ and set $P = C \cup D \cup E$. For these projects, we set $\ell_{c_i} = u_{c_i} = 3$, $\ell_{d_i} = u_{d_i} = 3$, and $\ell_{e_i} = u_{e_i} = 2$. 
	
	Next we define the preferences of the applicants, here we take the indices modulo $3m$, \ie $1-1 = 3m$ in our notation. Let $i \in [3m]$ and $T_{j}, T_{k}, T_{l}$ be the three sets that contain~$x_i$. For each $i$, we give applicant $b_i$ the preference list $c_{j} \succ c_{k} \succ c_{l} \succ d_i$, and $z_i$ gets the preference list $d_i \succ e_{i-1}$. For $Y$, we differentiate on the parity of $i$: if $i$ is odd, then we give $y_i$ the preference list $e_i \succ d_i$ and otherwise, if $i$ is even, then $y_i$ gets the preference list $d_i \succ e_i$---these applicants are the ones where we deviate from the proof of Cechl\'arov\'a and Fleiner~\cite{CF17}.
	
	We next show that an exact cover exists if and only if the matching $M$ defined by $M(d_i) = \lbrace b_i, y_i, z_i\rbrace$ for all $i \in [3m]$ is unpopular. 
	
	\subsubsection{$\Rightarrow$}
	Firstly assume that we have an exact cover $T_{i_1}, \dots, T_{i_m}$.
	Now we create a matching $\Mp$ with $\Mp(b_i) = c_{i_j}$ if $x_i \in T_{i_j}$, $\Mp(z_i) = e_{i-1}$ and $\Mp(y_i) = e_i$. This matching is feasible, since we have an exact cover and thus every $c_{i_j}$ is matched to exactly the three applicants from $B$ that could be matched to it.  
	Every $b_i$ and all $y_i$ with an odd index prefer $\Mp$ to $M$ while all $z_i$ and $y_i$ with an even index prefer $M$ to $\Mp$. Tallying the votes, it turns out that $\Mp$ is more popular than $M$, because $|B| = |Z|$, and since $m$ was odd, there are more $y_i$ applicants with an odd index than with an even index. 
	
	\noindent\subsubsection{ $\Leftarrow$} Next assume that there is a more popular feasible matching $\Mp$. If all $b_i$ are matched to some $c_{j}$ then we trivially have an exact cover. 
	Otherwise assume we have $b_{i_1}, \dots, b_{i_k}$,  who are not matched to a $c_j$ and assume that $i_1 < \dots < i_k.$
	We let $A_{i_j} \coloneqq \lbrace b_{i_\alpha}, y_{i_{\alpha}+1}, z_{i_\alpha} \mid i_j \le i_\alpha < i_{j+1}\rbrace $. The vote for $\Mp$ is simply the sum over all the votes in all the $A_{i_j}$. 
	If $\lvert A_{i_j} \rvert = 3$, then $b_{i_{j+1}}$ has to be unmatched as well. Thus either $b_{i_j}, z_{i_j}$, and $y_{i_{j+1}}$ all stay the same or $b_{i_j}, z_{i_j}$ both have to get worse. Therefore the vote of $A_{i_j}$ in this case is non-positive. 
	If $\lvert A_{i_j} \rvert > 3$, it is easy to see that 
	the sum over all votes except for $b_{i_j}, y_{i_j+1}$, and $z_{i_j}$ can be at most one, since the $y_{i_\alpha}$ are alternating and the $b_{i_j}$ and $z_{i_j}$ cancel out. 
	Furthermore it has to hold that either $b_{i_j}$ is matched to $d_{i_j}$ and thus $y_{i_{j}+1}$ is unmatched and votes negatively or $b_{i_j}$ is unmatched and $z_i$ also has to vote negatively. Therefore the vote in $A_{i_j}$ can be at most $0$ and thus $M$ could not have been popular. 	
\qed\end{proof}
\begin{restatable}{thm}{popconp}
\popular is \coNP-hard, even with $\lmax = 3 = \umax$, $\degA = 10$, and $\degH = 3$.
\end{restatable}
\begin{proof}
	We reduce from \textsc{x3c} again. 
	
	\subsubsection{Construction} We follow the notation from the last reduction and create three sets of applicants $B = \lbrace b_1, \dots, b_{3m} \rbrace$, $Y = \lbrace y_1, \dots, y_{3m} \rbrace$, $Z = \lbrace z_1, \dots, z_{3m} \rbrace$ and we set $A = B \cup Y \cup Z$. 
	Furthermore we create five sets of projects $C = \lbrace c_1, \dots, c_n\rbrace$, $C^\prime = \lbrace c^\prime_1, \dots, c^\prime_n\rbrace$, $C^{\prime\prime} = \lbrace c^{\prime\prime}_1, \dots, c^{\prime\prime}_n\rbrace$,  $D = \lbrace d_1, \dots, d_{3m} \rbrace$, $E = \lbrace e_1, \dots, e_{3m} \rbrace$ and set $P = C \cup C^\prime \cup C^{\prime\prime} \cup D \cup E$. 
	The projects in $D, C, C^\prime$, and $C^{\prime\prime}$ all have upper and lower quota $3$, while the projects in $E$ have upper and lower quota~$2$. 
	Next let $i \in [3m]$ and let $T_{j}, T_{k}, T_{l}$ be the three sets that contain~$x_i$. 
	Now for any $T_p \in \lbrace T_{j}, T_{k}, T_{l} \rbrace$ 
	let\[
	\mathcal C(T_p, x_i) = \begin{cases}
		c_{p}, c_{p}^\prime, c_{p}^{\prime \prime}, \text{if } x_i \text{ is the element with the smallest index in } T_p; \\
		c_{p}^\prime, c_{p}^{\prime \prime},c_{p},  \text{if } x_i \text{ is the element with the second smallest index in } T_p; \\
		 c_{p}^{\prime \prime},c_{p},c_{p}^\prime,  \text{if } x_i \text{ is the element with the third smallest index in } T_p; \\
	\end{cases} \]
and we give $b_i$ the preference list  $\mathcal{C}(T_{j}, x_i)\succ \mathcal{C}(T_{k}, x_i) \succ \mathcal{C}(T_{l}, x_i) \succ d_i$. Note that this preference structure falls under Observation~\ref{obs:cycl_low3}, enforcing each of the $T_p$ projects to be closed. The applicant $z_i$ gets the preference list $d_i \succ e_{i-1}$. The preference list of $y_i$ is set to $e_i \succ d_i$ for each odd $i$, while its is set to $d_i \succ e_i$ for each even~$i$. 

~\subsubsection{$\Leftarrow$} Firstly if no exact cover exists, then it follows by the calculations in Theorem~\ref{thrm:popv_lq3} 
that the matching $M$ with $M(d_i) = \lbrace b_i, y_i, z_i\rbrace$ is popular, since all we added were additional copies of projects in $C$.

~\subsubsection{$\Rightarrow$}
If however an exact cover exists, then the matching $M$ is again unpopular.
Now assume that we have a matching $\Mp \neq M$. 
If there is no $b_i$ that is matched to any $c_{j}$, $c^\prime_{j}$ or $c^{\prime\prime}_{j}$ in $M'$, then $M$ is more popular than $M^\prime$, since this would improve the matching for any $b_i$ and $z_i$ not matched to $d_i$. 
If there is however a $b_i$ that is matched to some (without loss of generality) $c_{j}$, then there need to be $b_k, b_l$ that are matched to $c_{j}$ as well. 
Now this instance falls under Observation~\ref{obs:cycl_low3} and thus $\Mp$ was not popular. 
\qed\end{proof}
\begin{restatable}{thm}{polqthr}
	\maxpareto is $\NP$-complete even when $\lmax = 3 = \umax$ and $\degA = 3 = \degH$.
	\label{thrm:po_lq3}
\end{restatable}
\begin{proof}
    First we observe that if a perfect matching exists, then a perfect Pareto optimal matching also must exist, because of the monotonicity of the dominance relation. 
    This implies that the problem is in \NP.
	With regard to the hardness we give a simple reduction from \textsc{x3c}.
	For each set $T \in \mathcal{T}$ we create a project $p_T$ with $\ell_{p_T} = u_{p_T} = 3$, and for each element $x \in X$ that appears in the sets $T_i, T_j, T_k$ we create an applicant $a_x$ with edges to projects $p_{T_i}, p_{T_j}, p_{T_k}$, where these three projects appear in an arbitrary order in the preference list of~$a_x$. 
	
	It is obvious that an exact cover exists if and only if a perfect matching exists, which is equivalent to the existence of a perfect Pareto optimal matching.
\qed\end{proof}

In order to prove Theorem~\ref{thm:fastalg}, we first describe a faster algorithm for weighted matching with lower quota~2.

While the theorem of Dudycz and Paluch~\cite{DP18} gives us an easy classification result, we show how to combine the methods of Cornuéjols~\cite{Cor88} and Dudycz and Paluch~\cite{DP18} to get a faster algorithm for our special case, which is also easier to implement. For this we first need to review some definitions of Dudycz and Paluch~\cite{DP18}. We start with the definition of \emph{neighboring type}, a measure to define similarity between graph factors (or matchings in our case). For this we are given an instance of the \textsc{general factor problem}. For any $v \in V$ and $k \in B_v$, let $u_k(v) \in B_v$ be the maximum element such that either $B_v \cap [k, u_k(v)] \subseteq 2 \mathbb{Z}$ or $B_v \cap [k, u_k(v)] \subseteq 2 \mathbb{Z} +1$, \ie $B_v \cap [k, u_k(v)]$ only contains elements of the same parity as~$k$. Similarly we let $\ell_k(v)$ to be the minimum element such that  $B_v \cap [\ell_k(v), k]$ has the same parity of~$k$. In the following we also call any solution to the \textsc{general factor problem} a $B$-matching. Given a $B$-matching $M$ we define $B_M(v) \coloneqq [\ell_{\lvert M(v) \rvert}(v), u_{\lvert M(v) \rvert }(v)]$. Using this we can finally define what a matching of neighboring type is.
\begin{definition}[Neighboring type, Dudycz and Paluch~\cite{DP18}]
Given two $B$-matchings $\Mp$ and $M$ let $D = \lbrace v \in V \mid \lvert \Mp(v) \rvert \notin B_M(v) \rbrace$. Now we say that $\Mp$ is of neighboring type of $M$ if one of 
\begin{enumerate}
    \item $\lvert D \rvert = 0$
    \item $\lvert D \rvert = 2$ and for $v \in D$, $B_{\Mp}(v)$ and $B_M(v)$ are \emph{adjacent}, meaning that either $\max(B_{\Mp}(v)) + 1 = \min(B_{\Mp}(v))$ or  $\max(B_M(v)) + 1 = \min(B_{\Mp}(v))$
    \item $\lvert D \rvert = 1$ and for $v \in D$, either $\min(B_M(v)) - \max(B_{\Mp}(v)) = 2$ or $\min(B_{\Mp}(v)) - \max(B_M(v)) = 2$
\end{enumerate}
holds.
\end{definition}
This now brings us to one of the key theorems (\cite[Theorem~2]{DP18}) of their paper, namely that while trying to find a matching of larger weight it is sufficient to look at matchings of neighboring types.
\begin{proposition}[Dudycz and Paluch~\cite{DP18}]
    Given a $B$-matching $M$ there exists a $B$-matching of greater weight than $M$ if and only if there exists a $B$-matching of greater weight and neighboring type of $M$.
\end{proposition}%
This theorem now has consequences for our matching problem. First we take a look at what kind of matchings have a neighboring type.
\begin{restatable}{obs}{obs:neighboring}
Given an instance with $\lmax = 2$ and a matching $M$, a matching $\Mp$ is of neighboring type to $M$ if there are 
\begin{enumerate}
    \item arbitrarily many projects $p \in P$ with $\ell_p = 2$, such that $\lvert M(p) \rvert = 2$ and $\lvert \Mp(p) \rvert = 0$ or $\lvert M(p) \rvert = 0$ and $\lvert \Mp(p) \rvert = 2$ and
    \item at most one of 
    \begin{enumerate}
        \item exactly one project $p \in P$ with $\lvert M(p)\rvert \neq 0$ and $\lvert \Mp(p) \rvert = \lvert M(p)\rvert + 2 $ or  $\lvert \Mp(p)\rvert \neq 0$ and $\lvert \Mp(p) \rvert = \lvert M(p)\rvert - 2 $;
        \item two of projects $p \in P$, such that $\lvert \Mp(p) \rvert = \lvert M(p)\rvert \pm 1 $ or applicants $a \in A$ that are matched in $M$ and unmatched in $\Mp$ or matched in $\Mp$ and unmatched in $M$.
    \end{enumerate}
\end{enumerate}
\label{obs:neigh_type}
\end{restatable}

\begin{proof}
    This mostly follows from the definition of neighboring type, the only way for the second case of the neighboring type definition to be feasible, is for an applicant to be matched/unmatched or for a project to get one less/more applicant matched to it. 
    
    Similarly the only way for the third case to happen, is for a project to gain two applicants, since the applicants only have a choice between $1$ and $0$ projects to be matched to them.
\qed\end{proof}
Now we want to turn this observation into an algorithm and for this we use the gadgets defined by~\cite{Cor88} to reduce our problem to the maximum weight matching problem without lower quotas.
For simplicity's sake we assume that each applicant has a unique last-resort with lower and upper quota $1$ project to which only she is connected with weight $0$, thus allowing us to assume that all applicants are matched in any matching. 
The structure of our algorithm is simple: we check all three cases induced by Observation~\ref{obs:neigh_type} and in cases 2a) and 2b) we iterate over all possible projects and applicants to change.
\begin{restatable}{lem}{lemalgneigh}
Given a weighted matching $M$ in an instance $\mathcal{I}$ of \wmatch with weight function $w$, there is an $\mathcal{O}(\lvert V \rvert\lvert E\rvert)$ time algorithm to check if there is a matching of neighboring type to $M$ fulfilling neither of 2a) and 2b) in Observation~\ref{obs:neigh_type}.
\end{restatable}%

\begin{proof}
    Based on Observation~\ref{obs:neigh_type}, the only changes we are allowed to make to the matching is to have arbitrarily many projects going from being matched to $2$ applicants to being matched to $0$ and vice versa. We now iteratively go over all projects with $\lvert M(p) \rvert \in \lbrace 0,2 \rbrace$ and replace them by a gadget, see Cornuéjols~\cite{Cor88} for this approach. In the end this will give us an instance $\mathcal{I}'$ of maximum weight perfect matching.
    
    Assume we are given a project $p \in P$ such that $p$ is adjacent to the applicants $a_1, \dots a_k$ with $k \ge 2$ and $\ell_p \ge 2$. Then we introduce the new vertices $p_1, \dots p_k$ and $s^p_1, \dots s^p_k$. We add an edge between $a_i$ and $p_i$ with weight $w(a_i,p)$  for all $i \in [k]$ and we connect $p_i$ to $s^p_j$ with weight $0$ for $i,j\in [k]$ ,\ie they are a complete bipartite graph. Finally we add an edge from $s^p_1$ to $s^p_2$ also with weight $0$.
    We now apply this transformation to all projects $p \in P$ with $\lvert M(p) \rvert \in \lbrace 0,2 \rbrace$. This gadget will ensure that $p$ can be either matched to $0$ or to $2$ applicants.

    Further for each project $p \in P$ with $\lvert M(p) \rvert \notin \lbrace 0,2 \rbrace$ such that $p$ is adjacent to the applicants $a_1, \dots a_k$ we replace $p$ by the new vertices $p_1, \dots p_{k}$ and $s^p_1, \dots s^p_{k - \lvert M(p) \rvert}$. We again connect $a_i$ to $p_i$ with weight $w(a_i,p)$ for $i \in [k]$ and connect all $p_i$ and $s^p_j$ in a fully bipartite manner with every edge having a weight of $0$. This gadget will ensure that $p$ can only be matched to $\lvert M(p) \rvert$ applicants.
    
    We now argue that every maximum weight perfect matching corresponds to a matching of neighboring type to $M$ with maximum weight, such that neither of 2a) and 2b) in Observation~\ref{obs:neigh_type} are fulfilled. 
    
    First assume that we are given a matching $N$ of neighboring type to $M$ and we want to create a corresponding perfect matching $\Mp$ in $\mathcal I'$. Now for each $p \in P$ with $\lvert M(p) \rvert \in \lbrace 0,2\rbrace$ and $\lvert N(p) \rvert = 0$ we match $p_i$ to $s^p_i$. For each $p \in P$ with $\lvert M(p) \rvert \in \lbrace 0,2\rbrace$ and $\lvert N(p) \rvert = 2$, such that $N(p) = \lbrace a_i, a_j \rbrace$, we match $s^p_1$ to $s^p_2$, $p_i$ to $a_i$ and $p_j$ to $a_j$ and all remaining $s^p_k$ and $p_k$ vertices in an arbitrary manner.
    
    For each $p \in P$ with $\lvert M(p) \rvert \notin \lbrace 0,2\rbrace$ it must hold that $\lvert N(p) \rvert = \lvert M(p) \rvert$. We can thus match any $a_i \in N(p)$ to $p_i$ and match the remaining $p_j$ arbitrarily to the $s^p_j$.
    
    It is easy to see that this matching is perfect and has the same weight as $N$. Of course we can also revert this process and given any perfect matching $\Mp$ in $\mathcal{I}'$ we can create a matching $N$ of the desired neighboring type of the same weight, by matching $a_i$ to the project $p$ in $N$ if $a_i$ is matched to $p_i$ in $\Mp$. The matching $N$ has the same weight as $\Mp$ and due to the structure of our gadget it is of the desired neighboring type. 
    
    Thus in order to verify whether $M$ is a maximum weight matching it is sufficient to test whether $\Mp$ is a maximum weight perfect matching, which is possible in linear time in the size of the graph, see~\cite[Section 26.2]{Schri03}.
    It is easy to see that the graph created in $\mathcal{I}'$ has $\mathcal{O}(\lvert V \rvert + \lvert E \rvert)$ vertices and $\mathcal{O}(\lvert E \rvert + \sum_{v \in V} \deg_v^2) \in \mathcal{O}(\lvert E \rvert \, \lvert V \rvert)$, see~\cite{deC98}.
\qed\end{proof}
The same construction also allows us to find a matching of neighboring type fulfilling either 2a) or 2b). 
\begin{corollary}
    Given a weighted matching $M$ in an instance $\mathcal{I}$ of \wmatch with weight function $w$, there is an $\mathcal{O}(\lvert V \rvert^3\lvert E\rvert)$ time algorithm to check if there is a matching of neighboring type to $M$ fulfilling either 2a) or 2b) in Observation~\ref{obs:neigh_type}.
\end{corollary}
\begin{proof}
    This can be done, by iterating over all possible projects in 2a) or all possible pairs of projects in 2b) and fixing their degree using the same gadget as in the previous theorem.
\qed\end{proof}
Using this we now have a faster algorithm for our polynomial-time solvable lower quota $2$ problems.

\fastalg*
\begin{proof}
    For \popver and \parver we can use Lemma~\ref{lem:pop_ver_maxm} and our previous algorithm to test whether the given matching is of maximum weight. For \maxpareto we can start with any perfect matching which we can find using $n$ iterations of finding a larger matching, followed by $\lvert E \rvert$ iterations of finding a matching that dominates our given matching. 
\qed\end{proof}

\popnpc*
\begin{proof}
    Note that the membership in $\NP$ follows from Corollary~\ref{cor:max_size_max_match}.
	We reduce from the \NP-complete problem of finding a popular matching in a non-bipartite graph (\cite{FKPZ19,GMSZ21}). 
	Given an instance $\mathcal{I}$ of the popular matching problem in the non-bipartite graph $G = (V,E)$ and corresponding preference lists, we create an instance $\mathcal{I}'$ of \popular by creating an applicant $a_v$ for each vertex $v \in V$ and a project $p_e$ with upper and lower quota 2 for each edge $e \in E$. If a vertex $v \in V$ has a preference list $v_1 \succ_v, \dots, \succ_v v_k$ in $\mathcal{I}$, then in $\mathcal{I}'$ we assign the applicant $a_v$ the preference list $p_{\{v_1, v\}} \succ_{a_v}, \dots, \succ_{a_v} p_{\{v_k, v\}}$, \ie applicant $a_v$ ranks all the edges in the same order that the vertex $v$ ranked them. 
	We now argue that $\mathcal{I}$ has a popular matching if and only if $\mathcal{I}'$ admits a popular matching. 	
	
	We define a straightforward bijection between feasible matchings in $\mathcal{I}$ and~$\mathcal{I}'$. To a given matching $\Mp$ in $\mathcal{I}'$ we create a matching $M$ 
	in $\mathcal{I}$ by including the edge $(u,v)$ if and only if $a_u$ and $a_v$ are matched to the same project in~$M'$. Conversely, for a given matching $M$ in $\mathcal{I}$ we create the matching $\Mp$ in $\mathcal{I}'$ with  $\Mp(a_u) = p_{\lbrace u,v\rbrace} = \Mp(a_v)$ if and only if $M(u) = v$. Notice that the construction guarantees that the vertex $v$ is matched to its $i$th choice in the matching $M$ in $\mathcal{I}$ if and only if the corresponding applicant $a_v$ in $\Mp$ in $\mathcal{I}'$ is matched to her $i$th choice. Thus if $M$ is popular in $\mathcal{I}$, then $\Mp$ was also popular in $\mathcal{I}'$ and vice versa. 
\qed\end{proof}

\subsection{Missing proofs from Section~\ref{sec:parcomp}}

\kerneln*
\begin{proof}
    To properly have a decision problem, we treat the maximum weight matching problem as the problem of deciding whether there is a matching with weight at least $x$.
    First if $\lvert P \rvert \le W2^n$, we can trivially return our current instance as a kernel of size $\mathcal{O}(W2^n)$.
    
    \subsubsection{Construction }Let $A' \subseteq A$ be any subset of applicants and for $w' \le W\lvert A' \rvert$ let \[P^{w'}_{A'} \subseteq P \coloneqq\lbrace p \in P \mid A' \subseteq N_p, \ell_p \le \lvert A' \rvert \le u_p, \sum_{a \in A'} w(a,p) = w' \rbrace,\]
    \ie the subset of projects, such that we could match all applicants of $A'$ to this project to get a weight of exactly $w'$. 
    Now for any $A'\subseteq A$ and  $w' \le W\lvert A' \rvert$ we put a marker on $n$ of the projects in $H^{w'}_{A'}$. If there are less than $n$ projects in $H^{w'}_{A'}$ we put a marker on all of them. 
    After this, we delete all projects without a marker. 
    
    \subsubsection{Correctness}
    If there was no matching with a weight of at least  $x$ in the original instance, then our reduced instance also does not admit a matching of weight at least $x$ since we only deleted projects.
    Now assume there was a matching $M$ with weight at least $x$ in the original instance. If all the opened projects in $M$ have a marker on them, then the matching is of course still feasible. If there is some project $p$ without a marker, then there has to be a different project that is unopened and has a marker on it, such that we can match all the applicants in $M(p)$ to this project with the same weight, since we guaranteed by our construction that there are at least $n$ of these projects. Thus we can iteratively transform $M$ into a feasible matching in our reduced instance.
    Finally we notice that each applicant can be adjacent to at most $2^n n W$ projects and thus we have at most $\mathcal{O}(2^n n^2 W)$ projects.
 \qed\end{proof}
\popnwone*
\begin{proof}
     We reduce from the W[1]-complete problem \textsc{multicolored independent set}. Here we are given a graph $G = (V, E)$ together with a partition of the vertices into $k$ color classes $V_1, \dots, V_k$ and we are searching for an independent set that includes exactly one vertex from each color class. To further simplify notation, we also assume that the graph is $p$-regular, each color class contains exactly $q$ vertices, and no two vertices in the same color class are adjacent (see~\cite{cygan2015parameterized} for this construction).
    \subsubsection{Construction}
    Our \popular instance consists of the following projects:
    \begin{itemize}
        \item For each $v \in V$ we include a \emph{vertex project} $p_v$ with lower and upper quota $4$.
        \item For each edge $e \in E$ we include an \emph{edge project} $p_e$ with lower and upper quota $4$.
        \item We create three \emph{dummy projects} $p_1, p_2, p_3$, each with lower and upper quota~$1$.
        \item  Finally for each $(c, d) \in [k]\times [k]$ we create a \emph{two-color project} $p_{c, d}$ with upper and lower quota $5$.
    \end{itemize} 
    Now we turn to the applicants. Let $c \in [k]$ be any color and let $v_1, \dots v_q$ be the vertices of~$V_c$. Further for any $v_i \in V_c$ let $e^i_1, \dots e^i_p$ be a list of the edges $v_i$ is adjacent to. We introduce two \emph{selection applicants} $a_c^1$ and $a_c^2$ with preferences 
    \begin{align*}
        a_c^1 \colon p_{e_1^1} \succ \dots\succ p_{e^1_p} \succ p_{v_1} \succ \dots \succ p_{e^q_1} \succ \dots\succ p_{e^q_p} \succ p_{v_q}, \\
        a_c^2 \colon p_{e_1^q} \succ \dots\succ p_{e^q_p} \succ p_{v_q} \succ \dots \succ p_{e^1_1} \succ \dots\succ p_{e^1_p} \succ p_{v_1}.
    \end{align*} 
Note that the preference lists of these selection applicants list the vertex projects in reverse order of each other and list all the edge projects belonging to the vertex projects before these vertex projects. 
Further for any color $c\in [k]$ we also add two \emph{enforcing applicants} $\hat a_c^1$ and $\hat a_c^2$ with preference lists 
\begin{align*}
    \hat a_c^1 \colon p_{v_1} \succ \dots \succ p_{v_q}\succ p_1 \succ p_2 \succ p_3 \succ p_{c, 1} \succ \dots   \succ p_{c, q},\\
    \hat a_c^2 \colon p_{v_q} \succ \dots \succ p_{v_1}\succ p_1 \succ p_2 \succ p_3 \succ p_{c, 1} \succ \dots  \succ p_{c, q}, 
\end{align*}
and we add two \emph{dummy applicants} $a_1$ and $a_2$ who both have the preference list $p_1 \succ p_2 \succ p_3$. 
Finally for each combination of colors $(c,d) \in [k] \times [k]$ we also add one last \emph{dummy applicant} $a_{c,d}$ whose preference list only consists of $p_{c, d}$. Thus the number of applicants is in $\mathcal{O}(k^2)$.

\subsubsection{$\Rightarrow$} First assume that $v_1, \dots, v_k$ form a multicolored independent set with $v_c \in V_c$ for $c \in [k]$. We now show that the matching $M$ with $M(p_{v_c})=  \lbrace a^1_c, a^2_c,\hat a_c^1,\hat a_c^2 \rbrace$ for  $c\in[k], M(p_1) = a_1$ and $M(p_2) = a_2$ is popular. 
To do so, we prove that any applicant improving in a matching has a corresponding applicant that gets worse in this matching.
\begin{itemize}
    \item First we notice that if $a_2$ would improve in any matching, then $a_1$ would get worse, thus canceling out their votes. 
    \item Similarly if any $\hat a_c^1 $ would improve, then $\hat a_c^2$ would get worse, due to the front part of their preference lists being reversed. Naturally, if $\hat a_c^2 $ would improve, $\hat a_c^1$ would get worse and thus their votes also cancel out in this case.
    \item If $a_c^1$ would improve by matching a vertex project, $a_c^2$ would get worse, thus canceling out their vote.
    \item Finally assume that $a_c^1$ could get better by matching an edge project $p_{v,u}$ with $v \in V_c$ and $u \in V_d$. Then, since we have an independent set, without loss of generality $a_d^1$ and $a_d^2$ cannot have been matched to $p_u$ in our original matching. Thus by the structure of our preference lists, at least one of $a_d^1$ and $a_d^2$ must have gotten worse. Further all of the applicants of $\hat a_c^1, \hat a_c^2, \hat a_d^1, \hat a_d^2$ must now be matched to a project they rank worse than their respective vertex project, while at most $a_{c,d}$ could have gotten better if all $5$ applicants got matched to $p_{c,d}$. Thus the total sum of their votes is at most $-1$.
\end{itemize} 
  Therefore since for any improving applicant, there has to be a unique applicant who gets worse, the matching $M$ is popular. 

 \subsubsection{$\Leftarrow$}
Next we assume that we have a popular matching $M$. 

\subsubsection{$a_{c,d}$ and $p_{c,d}$ must be unmatched} Firstly we show that for any $(c,d) \in [k]\times[k]$ the applicant $a_{c,d}$ must be unmatched. If $a_{c, d}$ is matched to $p_{c,d}$ then due to the lower quota of $5$, $\hat a_c^1, \hat a_c^2, \hat a_d^1, \hat a_d^2$ must all also be matched to $p_{c,d}$. Now if any of $a_c^1, a_c^2, a_d^1,$ and $a_d^2$ are unmatched, (without loss of generality $a_c^1$ and $a_c^2$) we can simply choose some $v \in V_c$ and match $a_c^1, a_c^2,\hat a_c^1,$ and $\hat a_c^2$ to $p_v$ thus improving the matching for $4$ applicants, while only making it worse for $3$, which implies that the matching was not popular. Similarly if $a_c^1, a_c^2, a_d^1,$ and $a_d^2$ are all matched, they must be matched to some edge projects. First, if they are matched to the same edge project $p_{\lbrace v,u \rbrace}$ with $ v\in V_c$ and $u \in V_d$, then we can take two vertices $v^\prime \neq v \in V_c$ and $u^\prime \neq u \in V_d$ and match $a_c^1, a_c^2,\hat a_c^1,$ and $\hat a_c^2$ to $p_{v^\prime}$ and  $a_d^1, a_d^2,\hat a_d^1,$ and $\hat a_d^2$ to $p_{u^\prime}$ thus improving the matching for exactly two of the four selection applicants $a_c^1, a_c^2,a_d^1, a_d^2,$ ,while the matching improves for $\hat a_c^1,\hat a_c^2,\hat a_d^1,$ and $\hat a_d^2$ and gets worse for $a_{c,d}$, thus implying that $M$ was unpopular. If the applicants are matched to different edge projects, we can do the same trick as earlier but incorporate the four applicants of the other colors as well. 

\subsubsection{No $\hat a_{c}$ can be matched to a dummy} We notice that the projects $p_1, p_2, p_3$ with their corresponding applicants induce an instance of Observation~\ref{obs:condorcet}. Thus in a popular matching only two of the dummy projects can be open, which in turn implies that the only applicants matched to the dummy projects can be the dummy applicants (otherwise the matching is not even Pareto optimal), therefore for any color $c \in [k]$ there must be exactly one vertex project $p_v^c$ open, to which $a_c^1, a_c^2,\hat a_c^1,$ and $\hat a_c^2$ are matched. 

\subsubsection{Constructing the independent set}
Let us now take the set $I \coloneqq \lbrace v \in V \mid p_v \text{ is open in } M \rbrace$ and assume that there is an edge between two vertices $v,u \in I$, with $v \in V_c$ and $u \in V_d$. Then if there is any edge between $u$ and $v$, we can match $a_c^1, a_c^2, a_d^1$ and $a_d^2$ to  $p_{\lbrace u,v \rbrace}$ and $\hat a_c^1,\hat a_c^2,\hat a_d^1,\hat a_d^2, $ and $a_{c,d}$ to $p_{c,d}$, thus improving the matching for $5$ applicants and making it worse for $4$ which in turn would imply that the matching was not popular. Therefore there cannot be edges between vertices in $I$ and thus $I$ must be a multicolored independent set.  
\qed\end{proof}

\popfptm*
\begin{proof}
    Our goal is to encode the matching instance in such a way that the resulting \textsc{parametric integer program} is feasible if and only if no popular matching exists. For this we choose the matrix $B$ and vector $d$ in such a way that all possible matchings can be represented by a vector $b$ satisfying $Cb \le d$. Further the matrix $B$ should be chosen in such a way that the vector $x$ represents a matching that is more popular than the matching represented by $b$. 
    
    \subsubsection{Notation}
    First, since we parameterize by $m$, each applicant $a \in A$ can be uniquely identified by her preference structure over $P$. There are at most $\mathcal{O}((m+1)!)$ different preference structures, hence 
    we can partition $A$ into $t \in \mathcal{O}((m+1)!)$ types. Let $A_1, \dots, A_t$ be this partition, such that any two applicants in the same set have the same preference list. Further we also refer to the projects that appear in the preference lists of applicants in $A_i$ as~$N_i$.
    Furthermore we slightly alter the notation of the previous sections to follow~\cite{BIM10} and define the vote of a vertex as follows.
    \[\vote_a(p_1, p_2) = 
    \begin{cases*}1, &if $p_1 \succ_a p_2$ \\
    0, &if  $p_1 = p_2$ \\
    -1, &if  $p_2 \succ_a p_1$
    \end{cases*} \] Following the definition of popularity, matching $\Mp$ is more popular than matching $M$ if and only if $\sum_{a \in A} \vote_a(\Mp(a), M(a)) \ge 1$ holds. For easier notation for any type $i \in  [t]$ we define $\vote_i(p_1, p_2) =\vote_a(p_1, p_2)$ where $a \in A_i$ is an applicant of type $i$. 
    
    \subsubsection{Construction of $C$}
    As noted earlier, our goal is to construct the linear program represented by $B$ and $d$ in such a way that every feasible matching is a solution to this linear program. 
    To reach this, for each $i \in [t]$ and $p \in N_i$ we create a variable $x_i^p$ that should indicate how many applicants of type $A_i$ are matched to project $p$. Moreover we add one variable $x_i^i$ indicating the number of unmatched applicants of type $i$. Furthermore for each project $p \in P$ we create a variable $o_p$ that should indicate whether project $p$ is open or closed. 
    This now leads us to the following linear program
\begin{alignat}{4}
        \displaystyle\sum\limits_{p\in N_i \cup \lbrace i \rbrace}   x_i^p &=& &\lvert A_i \rvert, \quad &\mbox{\ for each $i \in [t]$}\label{eq1ap}\\
        \displaystyle\sum\limits_{i \in [t] \colon p\in N_i} x_i^p - o_p u_p &\le& &0, \quad &\mbox{\ for each $p\in P$} \label{eq2ap}\\
        \displaystyle\sum\limits_{i \in [t] \colon p\in N_i} x_i^p - o_p \ell_p &\ge& &0,  \quad &\mbox{\ for each $p\in P$} \label{eq3ap}\\
        x_i^p &\ge& &0, \quad &\mbox{\ for each $ i \in [t]$ and $p\in N_i$}  \label{eq4ap}\\
        o_p &\in& &[0,1], \quad &\mbox{\ for each $p\in P$} \label{eq5ap}
\end{alignat}%
Here Constraint~(\ref{eq1ap}) enforces that all applicants are either matched or unmatched.
   Note that each feasible matching is a solution to this ILP. Currently, any solution to this ILP is of the form $(x_1^{p_1}, \dots, x_t^{p_m}, o_{p_1}, \dots,  o_{p_m})$. This, however, is not enough to fully model the ILP we need for~$B$. As a first step, for each variable of the form $x_i^p$ we add a second copy. Furthermore we add $2t$ variables $b_1, \dots, b_{2t}$ with $b_{2i} = \lvert A_i\rvert = b_{2i+1}$, then we add $2m$ variables that are forced to be $0$, and one variable that is forced to be $-1$. After this, each solution $b$ is of the form $(\lvert A_1\rvert, \lvert A_1 \rvert, \dots, \lvert A_t \rvert, \lvert A_t \rvert, x_1^{p_1},x_1^{p_1}, \dots, x_t^{p_m},x_t^{p_m}, o_{p_1}, \dots,  o_{p_m}, \underbrace{0, \dots, 0}_{2m \text{ times}}, -1)$.
   
   \subsubsection{Construction of $B$}
   We design $B$ to ensure that there is a matching $M'$ that is more popular than the matching $M$ induced by~$b$. For any type $i \in [t]$ and projects $p,p' \in N_i\cup\lbrace i \rbrace$, we create a variable $x_i^{p \rightarrow p'}$ that should indicate the number of applicants of type $i$ who were matched to project $p$ in $M$ and are matched to project $p'$ in $M'$. Furthermore, for each project $p \in P$, we again create a variable $o'_p$ indicating whether project $p$ is open or closed in $M'$.
   This now allows us to construct the final ILP.
   \begin{alignat}{4}
        \displaystyle\sum\limits_{p, p'\in N_i \cup \lbrace i \rbrace}   x_i^{p \rightarrow p'} &=& &\lvert A_i \rvert,  \quad &\mbox{ for each $i \in [t]$} \label{eq6ap}\\
         \displaystyle\sum\limits_{p' \in P\cup \lbrace i \rbrace}   x_i^{p \rightarrow p'} &=& &x_i^p , \quad &\mbox{ for each $i \in [t]$ and $p \in N_i$} \label{eq7ap}\\
        \displaystyle\sum\limits_{i \in [t], p \in P} x_i^{p \rightarrow p'} - o'_{p'} u_{p'} &\le& &0, \quad &\mbox{ for each $p' \in P$} \label{eq8ap} \\
       \displaystyle\sum\limits_{i \in [t], p \in P} x_i^{p \rightarrow p'} - o'_{p'} \ell_{p'} &\ge& &0, \quad &\mbox{ for each $p' \in P$} \label{eq9ap}\\
       \displaystyle\sum\limits_{i \in [t] p, p'\in N_i \cup \lbrace i \rbrace}  -\vote_i(p',p)x_i^{p \rightarrow p'} &\le& &-1 ,& \quad \label{eq10ap}\\
       x_i^{p \rightarrow p'} &\ge& &0, \quad &\mbox{ for each $i \in [t]$ and $p,p' \in N_i$} \label{eq11ap}\\
        o'_p &\in& &[0,1], \quad &\mbox{ for each $p \in P$} 
   \end{alignat}
   Here Constraint~(\ref{eq6ap}) ensures that all applicants have a new partner in $M'$, while Constraint~(\ref{eq7ap}) ensures that every applicant matched to some $p$ in $M$ now has a new partner. Constraints~(\ref{eq8ap}) and (\ref{eq9ap}) again ensure that the lower and upper quotas are met and finally Constraint~(\ref{eq10ap}) ensures that $M'$ is more popular than~$M$.
   
   \subsubsection{Correctness}
   For the correctness of our construction we first observe that our is indeed a \textsc{parametric integer program}. This follows from the fact that we can model any constraint of the type $\sum_i \alpha_i x_i = y$ as constraints $\sum_i \alpha_i x_i \le y$ and $\sum_i -\alpha_i x_i \le -y$. Thus the right hand-side of the second ILP is exactly $b$, except that we simply ignore the $o_p$ variables from~$b$. Now all that is left to show is that the \textsc{parametric integer program} is feasible if and only if no popular matching exists. First assume that the \textsc{parametric integer program} is feasible and let $M$ be any feasible matching in our instance. Then if we set $o_p = 1$ for any project $p \in P$ if and only if the project $p$ is open in $M$ and for any type $i \in [t]$ and projects $p \in N_i \cup \lbrace i  \rbrace$, we set $x_i^p$ to be the number of applicants of type $i$ matched to $p$ in $M$, the first ILP is feasible, since the first constraint is fulfilled by every applicant being matched, while the second and third constraints are fulfilled because $M$ adheres to the quotas.
      
   Now due to us having a feasible \textsc{parametric integer program} instance, there is an assignment of the variables to fulfill the second ILP. Using this assignment we can get a feasible matching $\Mp$, by opening all projects $p \in P$ with $o'_p = 1$ and matching $\sum_{p' \in P}x_i^{p' \rightarrow p}$ applicants of type $i$ to project $p$. Due to the Constraints~(\ref{eq6}), (\ref{eq8}), and~(\ref{eq9}) this implies that $\Mp$ is a feasible matching. Furthermore Constraint~(\ref{eq7}) gives us that we can create this matching from $M$ by matching exactly $x_i^{p \rightarrow p'}$ applicants of type $i$ that were matched to $p$ in $M$ to $p'$ in $\Mp$.
   And thus since $$\displaystyle\sum\limits_{i \in [t] p, p'\in N_i \cup \lbrace i \rbrace}  -\vote_i(p',p)x_i^{p \rightarrow p'} \le -1,$$ it holds that $$\displaystyle\sum\limits_{i \in [t] p, p'\in N_i \cup \lbrace i \rbrace}  \vote_i(p',p)x_i^{p \rightarrow p'} \ge 1,$$ 
    which implies that $\Mp$ is more popular, since $x_i^{p \rightarrow p'}$ is exactly the number of applicants who switched from $M$ to $\Mp$.
   
   Now assume that no popular matching exists and let $b$ be any feasible variable assignment to the first ILP. Then we can construct a feasible matching $M$, by opening all projects $p \in P$ with $o_p = 1$ and matching $x_i^p$ applicants of type $i$ to~$p$. Then there is a matching $\Mp$ that is more popular than $M$. By letting $x_i^{p \rightarrow p'}$ be the number of applicants of type $i$ who were matched to $p$ in $M$ and to $p'$ in $\Mp$ and by letting $o'_p = 1$ if and only if the project $p$ is open in $\Mp$, we get a feasible assignment to the second ILP, since we know that $$1 \le \sum_{a \in A} \vote_a(\Mp(a), M(a)) = \sum\limits_{i \in [t] p, p'\in N_i \cup \lbrace i \rbrace}  \vote_i(p',p)x_i^{p \rightarrow p'}.$$
   
   \subsubsection{Running time}
   It is easy to see that $k,\lvert B \rvert \in \mathcal{O}((tm)^2)$ and since the maximum entry of any matrix or vector is at most $n$, we get $\lvert\lvert B,C,d\rvert \rvert_\infty \in \mathcal{O}(n+m)$ and thus we can solve the \textsc{parametric integer program} in $\mathcal{O}(f(m)poly(n + m))$ for some computable function $f$, which shows that \popular is in FPT when parameterized by $m$.
\qed\end{proof}
\povhopen*
\begin{proof} 
	We show how to solve this problem via the \textsc{feasible flow with demands} problem. %
	First we assume that some $a \in A$ is given who 
	will be the designated applicant to receive a better partner in $\Mp$. 
	Now we create a directed bipartite graph 
	$G_a$ with vertices $A \cup \hopen \cup \{p_\perp,s,t\}$, where $p_\perp$ will represent the unmatched applicants.  We connect $s$ to all vertices in $A$ with demand and capacity of $1$ on the edge. 
	Next we connect $a$ to all projects in $\hopen$ that $a$ prefers to $M(a)$ and all applicants in $A \setminus \lbrace a \rbrace$ get connected to all projects in $\hopen$ which they prefer to $M(a)$ as well as to $M(a)$ if $M(a)$ is in $\hopen$, further if they are unmatched in $M$, we also connect them to $p_\perp$. 
	All of these edges have no demand and a capacity of $1$. Finally we connect every project $p \in  \hopen$ to $t$, each with a demand of $\ell_p$ and with a capacity of $u_p$ and we connect $p_\perp$ to $t$ with no demand and infinite capacity. 
	We solve the \textsc{feasible flow with demands} problem for all $a \in A$ in $G_a$ and if some $a \in A$ exists for which the demands are satisfiable, we return the respective induced matching as $\Mp$. That is, if a flow of $1$ goes from $a$ to $p$, we match $a$ to $p$ in $\Mp$ and if a flow of $1$ goes from $a$ to $p_\perp$ we leave $a$ unmatched in $\Mp$.
	Otherwise we return that our matching is Pareto optimal.

	First assume that there is a $G_a$ such that there is a feasible flow in $G_a$. Since each edge from $s$ to $A$ has a demand of $1$, this implies that each applicant in $A$ is either matched to a project in $\hopen$ 
	that she prefers to her project in $M$ or that she is matched to the same project in both matchings. Furthermore since $a$ is only connected to projects she prefers 
	to $M(a)$, this implies that $a$ is better off in $\Mp$. Finally since each project in $\hopen$ has a demand equal to its lower quota, each of the projects must be open and since its capacity is equal to its upper quota, the matching must be feasible. Thus $\Mp$ is a feasible matching and dominates $M$. 	
	
	Similarly if there is a matching $\Mp$ that only opens projects in $\hopen$ and dominates $M$, there has to be an applicant $a \in A$ who strictly prefers $\Mp$ to $M$. We can thus take the network $G_a$ and the flow that connects every $\hat a \in A$ to $\Mp(\hat a)$ if $\hat a$ is matched in $\Mp$ and connects it to $p_\perp$ if $\hat a$ is unmatched. Since the matching only opens projects in $\hopen$ and respects the quotas, this also leads to a flow respecting the demands incoming into $t$. Thus there is a feasible flow in~$G_a$. 
\qed\end{proof}
\begin{figure}[h]
\begin{minipage}{0.6 \textwidth}    \includegraphics[scale = 0.85]{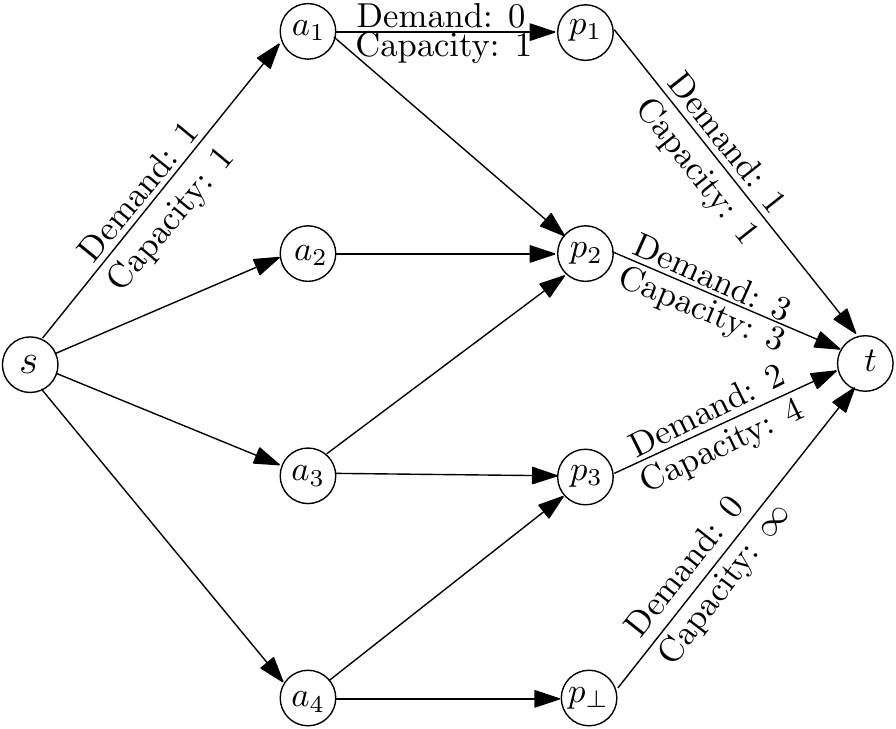}
\end{minipage}\hfill
\begin{minipage}{0.3 \textwidth}
    \begin{tabular}{rll}
       $\phantom{1111}\ell_{p_1} =$& 1 = &$u_{p_1}$\\
       $\ell_{p_2} =$& 3 = &$u_{p_2}$\\
       $\ell_{p_3} =$& 2, &$u_{p_3} = 4$\\
       $\ell_{p_4} =$& 2 = &$u_{p_4}$
    \end{tabular}
    \begin{align*}
       a_1 \colon& p_2 \succ p_1 \succ p_3 \\
       a_2 \colon& p_2 \succ p_4\\
       a_3 \colon& p_3 \succ p_2 \succ p_4 \\
       a_4 \colon& p_3
    \end{align*}
\end{minipage}
    \caption{Consider an instance of \parver with four applicants $a_1, \dots, a_4$ and four projects $p_1, \dots, p_4$ 
    with quotas and preference lists as shown on the right side of the figure.
    If we consider the matching $M$ with $M(a_1) = p_1, M(a_2) = p_4, M(a_3) = p_4$, and $M(a_4) = a_4$ with $\hopen = \lbrace p_1, p_2, p_3 \rbrace$, the above instance is the \textsc{feasible flow with demands} instance created in Theorem~\ref{thrm:pov_hopen}.} 
\end{figure} %
\maxwhopen*
\begin{proof}
	For our reduction to \textsc{max-cost circulation}  we again create a network with vertices $A \cup \hopen \cup \{p_\perp,s\}$, next we connect each applicant $a \in A$ to all projects she is adjacent to in the original graph, each with demand $0$, capacity $1$ and the same cost as in the original network and we connect $a$ to $p_\perp$, again with demand $0$, capacity $1$ and cost $0$, next we connect all projects $p \in  \hopen$ to $s$, each with demand $\ell_p$ and capacity $u_p$ and we connect $p_\perp$ to $s$ with demand $0$ and capacity $n$. Finally we connect $s$ to every applicant $a \in A$, each with demand and capacity $1$.

	It is obvious that matchings in our original graph in which exactly the projects in $\hopen$ are open, correspond to circulations in our new graph, since we enforce every project to have a circulation of at least the lower quota. Thus finding a maximum weight circulation is equivalent to finding a maximum weight matching that opens exactly the projects in $\hopen$. 
	
\qed\end{proof}
\fptmquota*
\begin{proof}
	In this proof we generalize the approach of Theorem~\ref{thrm:maxw_hopen}. We refer to $P_\text{quota}$ as the set of all projects with a lower quota greater than $1$.
	For this we create a network with vertices $\lbrace s,p_\perp \rbrace \cup A \cup P$. Now we iterate over all subsets $\hat P \subseteq P_\text{quota}$ and construct the same instance as in Theorem~\ref{thrm:maxw_hopen} with $\hopen = \hat P \cup P \setminus  P_\text{quota}$, but instead of connecting each project in  $P \setminus  P_\text{quota}$ with a demand of $1$ to $s$, we connect it with a demand of $0$. This guarantees that the projects in $P \setminus  P_\text{quota}$ can be either open or closed and a \textsc{max-cost circulation} in this instance corresponds to a maximum weight matching that opens all projects in $\hat P$ and can choose to open any number of projects in $P \setminus  P_\text{quota}$. Therefore by iterating over all possible projects with a non-unit lower quota, we can find the maximum weight matching. 
\qed\end{proof}

\woneparmopen*
\begin{proof}
	We reduce from the problem \textsc{exact unique hitting set} as defined in Section
	\ref{sec:prelims}. 
	 As our input we are given a set of elements 
	$X = \lbrace x_1, \dots, x_n\rbrace$, a set of subsets $\mathcal T \subseteq 2^X$ 
	
	For each element $x_i \in X$ let $T_i \subseteq \mathcal{T}$ be the set of sets in which the element $x_i$ appears. For each $x_i \in X$ we create a project $p_{x_i}$ with upper and lower quota $\lvert T_i \rvert$. 
	Further we create a \emph{last-resort project} $p^-$ with upper and lower quota $\lvert \mathcal T \rvert$. 
	Next for each $T = \lbrace \hat x_1, \dots, \hat x_\ell \rbrace \in \mathcal T$ we  create an applicant $a_T$ with preference list $p_{\hat x_1} \succ \ldots \succ p_{\hat x_\ell} \succ p^-$. Note that the ordering of the non-last-resort projects is arbitrary. 
	Now let $M$ be the matching that matches all applicants to $p^-$. 
	
	First we assume that there is an exact unique hitting set $H$ of size $k$ and construct a matching $M'$ that dominates~$M$. For each $T \in \mathcal T$ there is exactly one $x_i \in T \cap H$. We add $(a_T, p_{x_i})$ to our new matching~$M'$. Since the hitting set is unique, for every $x_i \in H$ the project $p_{x_i}$ meets their quota exactly. Thus $M'$ is feasible. Furthermore, $M'$ dominates $M$ because in $M'$ every applicant is matched to a project she prefers to~$p^-$. From this follows that $M$ was not Pareto optimal. 
		
	Next assume that there is a matching that opens exactly $k$ projects and dominates $M$. Let $H \subseteq X$ be the set of elements whose projects are open. Since the lower quota of each element is exactly the number of sets this element appears in, each set must be covered by exactly one element of $H$. Thus $H$ must be an exact unique hitting set. 
	
	For \popver we can mostly follow the reduction, except that we set the upper and lower quota of $p^-$ to $2n-1$ and we add $n-1$ dummy applicants, each of whom only has $p^-$ in her preference list. Following the same argumentation as in the previous proof it is easy to see that the matching that matches all applicants to $p^-$ is popular if and only if no exact unique hitting set exists, since else we can construct a matching 
	for which there is one more applicant who prefers the new matching to $M$, than there are applicants who prefer $M$ to the new matching, by matching all element projects to their respective sets as in \parver case. %
\qed\end{proof}
\wonepopmopen*
\begin{proof}
	We modify the previous reduction and reduce from \textsc{x3c} instead of \textsc{exact unique hitting set}.
	We are again given a set of items $X = \lbrace x_1, \dots, x_{3m} \rbrace$ and a set $\mathcal T = \lbrace T_1, \dots, T_k  \rbrace\subseteq 2^X$ 
	of size $3$ subsets and we want to decide whether there is a partition $T \subseteq \mathcal T$ of $X$. First for each $T_i \in \mathcal{T}$ we check if there is another $T_j \in \mathcal{T}$ such that $T_i \cap T_j = \emptyset$. If there is no such set $T_j$ and $\lvert X \rvert > 3$, we can simply remove $T_i$ from $\mathcal T$ as it cannot be part of any partition, since no other set could be included. 
	
	After this we follow the previous reduction and add one project $p_i$ for each $T_i \in \mathcal T$ with upper and lower quota $3$ and a last-resort project $p^-$ with lower and upper quota $6m-1$. 
	We again add one applicant $a_x$ per element $ x \in X$ with preference list $p_{i_1} \succ_{a_x} p_{i_2} \succ_{a_x} p_{i_3} \succ_{a_x} p^-$, where the projects correspond to the three sets $T_{i_1}, T_{i_2}$, and $T_{i_3}$ in which $x$ appears. 
	We also add $3m-1$ dummy applicants, each of whom only has $p^-$ on their preference list. Exactly as in the previous proofs, the matching that assigns every applicant to $p^-$ is popular if and only if an exact cover exists. Furthermore any matching opening only one of the projects corresponding to sets cannot be popular, since we can always open a second project with different applicants, since we required that any set has another that is disjoint from it. 
	Thus there is a popular matching with exactly one open project if and only if there is no exact cover.  
	 The result for \maxpareto simply follows by removing the dummy agents.
\qed\end{proof}

\mapomclosed*
We again split this theorem into four separate statements in the appendix for better readability. 

\begin{restatable}{thm}{mapom_closed_appendix}
	 Given an instance $\mathcal{I}$ of \maxpareto and parameter $\mclosed$, it is $\wone$-hard to decide whether a perfect Pareto optimal matchings exists that closes exactly $\mclosed$ projects.
    \label{thrm:mapom_closed_appendix}
\end{restatable}
\begin{proof}
    We reduce from the \textsc{multicolored independent set} Problem. Here we are again given a graph $G = (V,E)$ together with a partition of the vertices into $k$ color classes $V_1, \dots, V_k$ and our goal is to find an independent set with exactly one vertex per class. We again assume that each color class has exactly $q$ vertices.  
    
    \subsubsection{Construction} 
    First we introduce the projects we add to our construction.
    \begin{itemize}
        \item For each $v \in V$ we create one \emph{vertex project} $p_v$ with lower and upper quota $q$.
        \item For each edge $e\in E$ we create one \emph{edge project} $p_e$ with lower and upper quota $2q+1$.
        \item For each combination of colors $(c,d) \in [k]\times[k]$ with $c \neq d$ we create a \emph{two-color project} $p_{c,d}$ with lower and upper quota $1$.
    \end{itemize}
     Now we turn to the applicants.
     \begin{itemize}
         \item For each color $c \in [k]$ and pair of vertices $(v,u) \in V_k^2$ with $v \neq u$ we create an applicant $a_{v,u}$ with preference list $p_v \succ p_u$.
         \item For each $e = \lbrace u,v \rbrace \in E$ with $v \in V_c$ and $u \in V_d$ we add $2q$ \emph{edge applicants} $a_e^1, \dots, a_e^{2q}$ each with preference list $p_{u} \succ p_v \succ p_e$ and one \emph{dummy applicant} $a_e$ with preference list $p_{c,d} \succ p_e$.
     \end{itemize} 
    Finally we set $\mclosed \coloneqq k + {k \choose 2}.$
    
    \subsubsection{$\Rightarrow$}    
    \subsubsection{Matching construction}
    First we assume that we have a multicolored independent set $v_1, \dots, v_k$ with $v_c \in V_c$ for $c \in [k]$. 
    We now construct a perfect Pareto optimal matching $M$ by matching 
    \begin{itemize}
        \item for any edge $e \in E$ all edge applicants to $p_e$;
        \item for any color $c \in [k]$ and $v \in V_c \setminus \lbrace v_c \rbrace$  all applicants in $\lbrace a_{v, u} \mid u \in V_c\setminus \lbrace v \rbrace\rbrace$ and  $a_{v_c, v}$ to $p_v$.
    \end{itemize}
    It is easy to see that $M$ closes exactly the projects $p_{v_1}, \dots, p_{v_k}$ and all two-color projects, thus closing exactly $\mclosed = k + {k \choose 2}$ projects.
    Also since all applicants are matched, $M$ is perfect. 
    
    \subsubsection{Pareto optimality}
    In order to show that $M$ is Pareto optimal, we go over all applicants and show that they cannot improve without any other applicant getting worse. 
    \begin{itemize}
        \item  For any color $c \in [k]$ and any vertices $v,u \in V_c$ with $v \neq v_c$ the applicant $a_{v,u}$ cannot improve since she is matched to her top choice in $M$.
        \item For the applicant $a_{v_c, v}$ the only possibility to improve is to open the project $p_{v_c}$, but the only applicants who prefer this project to their current project are all $q$ projects in $\lbrace a_{v_c, v} \mid v \in V_c \setminus \lbrace v_c \rbrace\rbrace$ and all edge projects corresponding to edges adjacent to~$v_c$.
        Therefore if we match $a_{v_c, v}$ to $p_{v_c}$, then at least $1$ but at most $q-1$ applicants corresponding to any edge can be matched to this project. This leaves at least one edge, such that $q+1$ of their edge applicants need a new project. But since their only possibility for this would be the other vertex project corresponding to the other vertex of the edge which has an upper quota of $q$, this would imply that at least one applicant would get worse compared to $M$. Therefore the applicant $a_{v_c}$ cannot improve. 
        \item Let $e  = \lbrace u,v \rbrace \in E$, and without loss of generality assume that $a_e^1$ should improve. Then they and all other $2q-1$ applicants should be either matched to $p_v$ or $p_u$. Since $e$ is an edge, either $u$ or $v$ is not in the independent set and therefore there is at least one top choice applicant matched to the corresponding project that would be unmatched in case of all $2q$ edge applicants being matched. 
    \end{itemize}
    Thus no applicant can improve and $M$ is Pareto optimal. 
    
    \subsubsection{$\Leftarrow$}
    Now assume that we have a perfect Pareto optimal matching $M$ with exactly $k + {k \choose 2}$ closed projects. First, since $M$ is perfect (so every applicant needs to be matched) at most $k$ vertex projects can be closed. Next we notice that there are exactly ${k \choose 2}$ two color projects. 
    \begin{itemize}
        \item Since for any closed edge project $p_e$, the applicant $a_e$ would need to be matched to the corresponding two color project and the applicants $a_e^1, \dots, a_e^{2q}$ would need to be matched to a vertex project with no vertex applicants. This implies that in any perfect matching with $k + {k \choose 2}$ closed projects, all edge projects must be open since otherwise two of the $k + {k \choose 2}$ possible closed projects would be open.
        \item Due to the perfect requirement there can be at most one vertex project $p_{v_c}$ closed per color $c \in [k]$.
        \item Because we have a Pareto optimal matching there cannot be any two closed projects $p_{v_c}$ and $p_{v_d}$ with $c, d\in [k]$ with an edge $e$ between them, since otherwise matching $a_e$ to $p_{c,d}$, $a_e^1, \dots, a_e^q$ to $p_{v_c}$ and $a_e^q, \dots, a_e^{2q}$ to $p_{v_d}$ would dominate $M$.
    \end{itemize} Therefore the set of vertices corresponding to the closed projects in $M$ is a multicolored independent set.
\qed\end{proof} 
\begin{restatable}{thm}{domclosed}
	Given an instance $\mathcal{I}$ of \parvernsp, parameter $\mclosed$, and matching $M$, it is $\wone$-hard to decide whether matching exists that closes exactly $\mclosed$ projects and dominates~$M$.
    \label{thrm:dom_closed}
\end{restatable}
\begin{proof}
    We reduce from the multicolored clique problem. For this we are again given a graph $G= (V,E)$ with a partition into color classes $V_1, \dots, V_k$ and our goal is to find a clique adhering to this partition. For simpler notation we assume that the graph induced by each color class is an independent set. This allows us to use clique and multicolored clique interchangeably and we can ensure that edges are always between two different colors in our construction. 
    
    \subsubsection{Construction}
    First as our projects we include 
    \begin{itemize}
        \item  for each $v \in V$, a \emph{vertex project} $p_v$ with lower and upper quota $k-1$;
        \item  for each edge $e \in E$, an \emph{edge project} $p_e$ with lower and upper quota $2$;
        \item  for each color $c \in [k]$, a \emph{color project} $p_c$ with lower and upper quota $k-1$. 
    \end{itemize}
    Our applicants will be the following.
    \begin{itemize}
        \item For each color $c \in [k]$ and vertex $v \in V_c$, we add $k-1$ \emph{vertex applicants} $a_v^1, \dots, a_v^{k-1}$, each with preference list $p_c \succ p_v$.
        \item  For each edge $e = \lbrace u,v \rbrace \in E$, we include two \emph{edge applicants}, $a_e^u$ with preference list $p_u \succ p_e$ and $a_e^v$ with preference list $p_v \succ p_e$.
    \end{itemize} 
    Our matching $M$ now matches all vertex applicants to their corresponding vertex project, \ie it matches all applicants in $\lbrace a_v^1, \dots, a_v^k\rbrace$ to $p_v$, and it matches all edge applicants to their edge project.
    Finally we set $\mclosed \coloneqq {k \choose 2}$ 
    
    \subsubsection{$\Rightarrow$}
    Let us assume we have a clique $C = \lbrace v_1, \dots, v_k\rbrace$ with $v_c \in V_c$ for $c \in [k]$.
    Then we take the matching $\Mp$ such that
    \begin{itemize}
        \item For any color $c\in[k]$ the vertex applicants $a_{v_c}^1, \dots, a_{v_c}^{k-1}$ are matched to $p_c$, thus improving over $M$ for all of them. 
        \item For any vertex $v \in V \setminus C$, \ie $v$ is not in the clique, we match the vertex applicants $a_{v}^1, \dots, a_{v}^{k-1}$ are matched to $p_v$.
        \item For any edge $e = (v,u)$ such that $v$ and $u$ are in the clique, we match $a_e^v$ to $p_v$ and $a_e^u$ to $p_u$, thus improving their matching.
        \item For any edge $e = (v,u)$ such that at least one of $v$ and $u$ is not in the clique, we match $a_e^v$ and $a_e^u$ to $p_e$.
    \end{itemize}
    The matching $\Mp$ does not make any applicant worse, adheres to the lower and upper quotas, due to each vertex in the clique having exactly $k-1$ neighbors in the clique and closes exactly the ${k \choose 2}$ projects corresponding to the edges between the vertices in the clique. 
    
    \subsubsection{$\Leftarrow$}
    Assume that we have a matching dominating $M$, which closes exactly ${k \choose 2}$ projects. It is easy to see that the only way to close ${k \choose 2}$ projects while simultaneously matching all applicants is to close ${k \choose 2}$ edge projects, matching the edge applicants to the corresponding vertex projects and the vertex applicants to the color projects. This however implies that all vertices corresponding to promoted vertex projects must have an edge to all the other vertices, thus forming a multicolored clique.  
\qed\end{proof}
\begin{restatable}{cor}{popvmclosed}
Given an instance $\mathcal{I}$ of \popvernsp, parameter $\mclosed$, and matching $M$, it is $\wone$-hard to decide whether matching exists that closes exactly $\mclosed$ projects and is more popular than~$M$.
    \label{corr:popv_mclosed}
\end{restatable}
\begin{proof}
Instead of having $k$ color projects, we just have one selection project $p_s$ with lower quota $k(k-1)$ which is ranked lowest by all vertex projects and we add one selection applicant $a_s$ who only ranks $p_s$. Then the matching $M$ from the last proof 
is obviously less popular than the matching $\Mp$ assigning 
all the edge applicants as in the last proof, while all the vertex applicants and the selection applicant to the selection project. 
Furthermore since the votes of vertex and edge projects cancel each other out, the only way to get a more popular matching closing exactly ${k \choose 2}$ projects is to match the selection gadget and also $k(k-1)$ edge applicants, which just as in the last proof implies the existence of a clique.
\qed\end{proof}
\begin{restatable}{thm}{popmclosed}
	 Given an instance $\mathcal{I}$ of \popular and parameter $\mclosed$, it is $\wone$-hard to decide whether a popular matching exists that closes exactly $\mclosed$ projects, even if $\mathcal{I}$ does 
    not admit any other popular matching.
    \label{thrm:pop_mclosed}
\end{restatable}
\begin{proof}
    We again reduce from the multicolored clique problem, where we are given a graph $G = (V,E)$ with partition $V_1, \dots, V_k$ of the vertices of~$G$.
    
    \subsubsection{Construction}
    As our main projects we include
    \begin{itemize}
        \item for each color  $c \in [k]$  a \emph{color project} $p_c$ with lower and upper quota $k$;
        \item for each vertex $v \in V$ a \emph{vertex project} $p_v$ with lower and upper quota $k$;
        \item for each vertex $v \in V$ a \emph{vertex project} $p_v$ with lower and upper quota $k$;
        \item for each edge $e\in E$ an \emph{edge project} $p_e$ with lower and upper quota $2$.
    \end{itemize}
    Further as our \emph{dummy projects} we include
    \begin{itemize}
        \item three \emph{dummy color projects} $p_{C,1}, p_{C,2}, p_{C,3}$;
        \item six \emph{dummy vertex projects}  $p_{V,1}, p_{V,2}, p_{V,3}, p_{V,4}, p_{V,5}, p_{V,6}$;
        \item three \emph{dummy edge projects} $p_{E,1}, p_{E,2}, p_{E,3}$,
    \end{itemize}
     with all dummy projects having a lower and upper quota of $1$. 
     
     Our non-dummy applicants are as follows.
     \begin{itemize}
         \item For each color $c\in[k]$ and vertex $v \in V_c$ we add $k-1$ \emph{vertex applicants} $a_v^1, \dots, a_v^{k-1}$ with preference list $p_c \succ p_v \succ p_{V,4} \succ p_{V,5} \succ p_{V,6}$.
         \item For each edge $e = \lbrace u,v \rbrace$ we add two \emph{edge applicants} $a_e^u$ with preference list $p_u \succ p_e \succ p_{E,1} \succ p_{E,2} \succ p_{E,3} $ and $a_e^v$ with preference list $p_v \succ p_e \succ p_{E,1} \succ p_{E,2} \succ p_{E,3} $.
     \end{itemize}
     Finally our dummy applicants will be as follows.
     \begin{itemize}
         \item For each color $c\in[k]$ and vertex $v \in V_c$ we add one \emph{vertex dummy  applicant} $a_v^D$ with preference list $p_v \succ p_{V,1} \succ p_{V,2} \succ p_{V,3}$.
         \item For each color $c \in [k]$ we add one \emph{color dummy applicant} $a_c^D$ with preference list $p_c \succ p_{C,1} \succ p_{C,2} \succ p_{C,3}$.
         \item We add two \emph{common color dummy applicants} $a_{C,1}, a_{C,2}$ with preference list $p_{C,1} \succ p_{C,2} \succ p_{C,3}$.
         \item We add two \emph{common vertex dummy applicants} $a_{V,1}, a_{V,2}$ with preference list $p_{V,1} \succ p_{V,2} \succ p_{V,3}$.
         \item We add two \emph{common enforcing vertex dummy applicants} $a_{V,3}, a_{V,4}$ with preference list $p_{V,4} \succ p_{V,5} \succ p_{V,6}$. 
         \item Finally we add two \emph{common edge dummy applicants} $a_{E,1}, a_{E,2}$ with preference list $p_{E,1} \succ p_{E,2} \succ p_{E,3}$.
     \end{itemize}
     This construction heavily relies on Observation~\ref{obs:condorcet}. The vertex dummy applicants guarantee that in each popular matching each vertex project must be open and the dummy color applicants guarantee that each color project must be open. Similarly the common enforcing vertex dummy applicants force every vertex applicant to be matched and the common edge dummy applicants force every dummy applicant to be matched.
    Now we set $\mclosed \coloneqq {k \choose 2} + 3.$
    
    \subsubsection{$\Rightarrow$}
    First assume that we have a clique $ C = \lbrace v_1, \dots, v_k \rbrace$, with $v_c \in V_c$ for $c \in [k]$.
    Now we construct a matching in the following way.
    \begin{itemize}
        \item For each $v_c \in C$ we match all $k-1$ vertex applicants of $v_c$ as well as $a_c^D$ to $p_c$.
        \item Furthermore for each $v_c \in C$ and any edge $e = \lbrace v_c, v_i \rbrace$ with $v_i \in C$, we match $a_e^{v_c}$ to $p_v$. Adding on to this we also match $a_{v_c}^D$ to $p_v$.
        \item For each $v \in V \setminus C$ we match all vertex applicants as well as the vertex dummy applicant belonging to $v$ to $p_v$.
        \item For each edge, where at least of the two end vertices is not in the clique we match the edge to the corresponding edge project.
        \item Finally turning to the common dummy applicants for $* \in \lbrace C, V, E\rbrace$ we match $a_{*,1}$ to $p_{*,1}$ and $a_{*,2}$ to $p_{*,2}$.
    \end{itemize}
    
    \subsubsection{Popularity}
    Note that the feasibility of the matching immediately follows from the fact that we have a clique of size $k$. If any vertex applicant improves by being matched to a color project, then there has to be at least one other vertex applicant being unmatched from the color project, thus canceling out their votes. Similarly if any edge applicant matched to an edge project wants to improve, either a vertex applicant needs to be unmatched or an edge applicant previously matched to a vertex applicant now needs to be matched to an edge applicant, also canceling out their vote. The last type of applicant not matched to their top choice are three of the common dummy applicants, but for them to improve the other common dummy applicant would need to get worse, thus also canceling out their vote. Therefore the matching is popular. 
    
    \subsubsection{$\Leftarrow$}
    Next assume that we have a popular matching. Then the common dummy applicants enforce that each color and vertex project must be filled, each vertex project must be matched and each edge project must be matched to either a vertex or an edge project. Thus there must be $k$ vertices $v_1, \dots, v_k$ matched to a color project. Let us assume that vertices $v_c$ and $v_d$ for $c \neq d \in [k]$ do not have an edge between each other. As observed earlier the dummy applicants enforce that the vertex projects must be open, and thus there must be at least one edge applicant $a_e^v$ matched to $p_{v_c}$ such that other the end vertex of $e$ is not in $v_1, \dots, v_k$. But since all dummy applicants and vertex applicants must be matched, and therefore all vertex projects must already be full, the edge applicant must be unmatched and thus the matching would be unpopular.
    Therefore $v_1, \dots, v_k$ must be a clique.
\qed\end{proof}

\end{document}